\documentclass[conference]{IEEEtran}
\makeatletter
\def\ps@headings{%
\def\@oddhead{\mbox{}\scriptsize\rightmark \hfil \thepage}%
\def\@evenhead{\scriptsize\thepage \hfil \leftmark\mbox{}}%
\def\@oddfoot{}%
\def\@evenfoot{}}
\makeatother
\usepackage{enumerate}
\usepackage{epsfig}
\usepackage{graphicx}
\usepackage{amsmath}
\usepackage{amssymb}
\usepackage{url}
\usepackage{subfigure}
\usepackage{wrapfig}

\newtheorem{lemma}{\bf{Lemma}}
\newtheorem{thm}{\bf{Theorem}}

\def\profilr{${\textrm{\textsc{Profil}}_R}~$}

\def\spotrv{${\textrm{\textsc{Spotr}}_V}~$}

\newenvironment{icompact}
{
\begin{list}
        {\labelitemi}
        {\leftmargin=0em \itemindent=1em \itemsep=1pt \parskip=0pt \parsep=0pt}
}
{\end{list}}

\begin{document}

\title{Eat the Cake and Have It Too: Privacy Preserving Location Aggregates in Geosocial Networks}

\author{
\IEEEauthorblockN{Bogdan Carbunar, Mahmudur Rahman, Jaime Ballesteros, Naphtali Rishe}
\IEEEauthorblockA{School of Computing and Information Sciences\\
Florida International University\\
Miami, Florida 33199\\
Email: \{carbunar, mrahm004, jball008, rishen\}@cs.fiu.edu}
}

\maketitle

\begin{abstract}
Geosocial networks are online social networks centered on the locations of
subscribers and businesses. Providing input to targeted advertising, profiling
social network users becomes an important source of revenue. Its natural
reliance on personal information introduces a trade-off between user privacy
and incentives of participation for businesses and geosocial network providers.
In this paper we introduce \textit{location centric profiles} (LCPs),
aggregates built over the profiles of users present at a given location.  We
introduce \profilr, a suite of mechanisms that construct LCPs in a private and
correct manner. We introduce iSafe, a novel, context aware public safety
application built on \profilr. Our Android and browser plugin implementations
show that \profilr is efficient: the end-to-end overhead is small even under
strong correctness assurances.
\end{abstract}

\section{Introduction}

Online social networks have become a significant source of personal
information.  Facebook alone is used by more than 1 out of 8 people today.
Social network users voluntarily reveal a wealth of personal data, including
age, gender, contact information, preferences and status updates. A recent
addition to this space, geosocial networks (GSNs) such as Yelp~\cite{Yelp},
Foursquare~\cite{foursquare} or Facebook Places~\cite{FBPlaces}, further
provide access even to personal locations, through \textit{check-ins}
performed by users at visited venues.

From the user perspective, personal information allows GSN providers to offer
targeted advertising and venue owners to promote their business through
spatio-temporal incentives (e.g., rewarding frequent customers through
accumulated badges). The profitability of social network providers and
participating businesses rests on their ability to collect, build and
capitalize upon customer and venue profiles.  Profiles are built based on user
information -- the more detailed the better.
Providing personal information exposes however users to significant risks, as
social networks have been shown to leak~\cite{KW10} and even
sell~\cite{FBPrivacy} user data to third parties. Conversely, from the provider
and business perspective, being denied access to user information discourages
participation.  There exists therefore a conflict between the needs of users
and those of providers and participating businesses: Without privacy people may
be reluctant to use geosocial networks, without feedback the provider and
businesses have no incentive to participate.

In this paper we take first steps toward breaking this deadlock, by introducing
the concept of \textit{location centric profiles} (LCPs). LCPs are
aggregate statistics built from the profiles of (i) users that have visited a
certain location or (ii) a set of co-located users.

We introduce \profilr, a framework that allows the construction of LCPs based
on the profiles of present users, while ensuring the privacy and correctness of
participants. Informally, we define privacy as the inability of venues and the
GSN provider to accurately learn user information, including even anonymized
location trace profiles. Thus, location privacy is an inherent \profilr
requirement.

Correctness is a by-product of privacy: under the cover of privacy users may
try to bias LCPs. We consider two correctness components (i) location
correctness -- users can only contribute to LCPs of venues where they are
located and (ii) LCP correctness -- users can modify LCPs only in a predefined
manner.  Location correctness is an issue of particular concern. The use of
financial incentives by venues to reward frequent geosocial network customers,
has generated a surge of fake check-ins~\cite{cheat1}.  Even with GPS
verification mechanisms in place, committing location fraud has been largely
simplified by the recent emergence of specialized applications for the most
popular mobile eco-systems (LocationSpoofer~\cite{LocationSpoofer} for iPhone
and GPSCheat~\cite{GPSCheat} for Android).

We propose first a venue centric \profilr. To relieve the GSN provider from a
costly involvement in venue specific activities, \profilr stores and builds
LCPs at venues. Participating venue owners need to deploy an inexpensive device
inside their business, allowing them to perform LCP related activities
\textit{and} verify the physical presence of participating users. We extend
\profilr with the notion of \textit{snapshot} LCPs, built by user devices from
the profiles of co-located users, communicated over ad hoc wireless
connections.  Snapshot LCPs are not bound to venues, but instead user devices
can compute LCPs of neighbors at any location of interest.  \profilr relies on
(Benaloh's) homomorphic cryptosystem and zero knowledge proofs to enable
oblivious and provable correct LCP computations.

We further introduce iSafe, a context aware safety application, that uses
\profilr to privately build safety LCPs. The constant population density
increase, and the recent surge of natural and man-made disasters, riots and
lootings, make safety aware applications of paramount importance.  The goal of
iSafe is to make users aware of the safety of their surroundings while
preserving the privacy of participants. Safety information can empower a suite
of applications, including safe walking/evacuation directions and safety
dependent mobile authentication.

We implemented iSafe and \profilr as mobile application and browser plugin
components.  Our experiments show that on a smartphone, with a client cheating
probability of 1 in a million, the end-to-end overhead of an LCP update
operation is 2.5s.  We further rely on data collected from Yelp~\cite{Yelp}, a
popular geosocial network, to build user and venue safety labels. The iSafe
browser plugin introduces an overhead of under 1s for collecting and processing
500 Yelp reviews.

The paper is organized as follows. Section~\ref{sec:model} describes the system
and adversary model and defines the problem. Section~\ref{sec:profilr}
introduces \profilr and proves its privacy and correctness.
Section~\ref{sec:realtime} introduces snapshot LCPs and presents the
distributed, real-time variant of \profilr.  Section~\ref{sec:isafe} introduces
iSafe and its implementation.  Section~\ref{sec:eval} evaluates the performance
of the proposed constructs.  Section~\ref{sec:related} describes related work
and Section~\ref{sec:conclusions} concludes.

\section{Background and Model}
\label{sec:model}

\begin{figure}
\centering
\subfigure[]
{\label{fig:distrev}{\includegraphics[width=1.63in,height=1.35in]{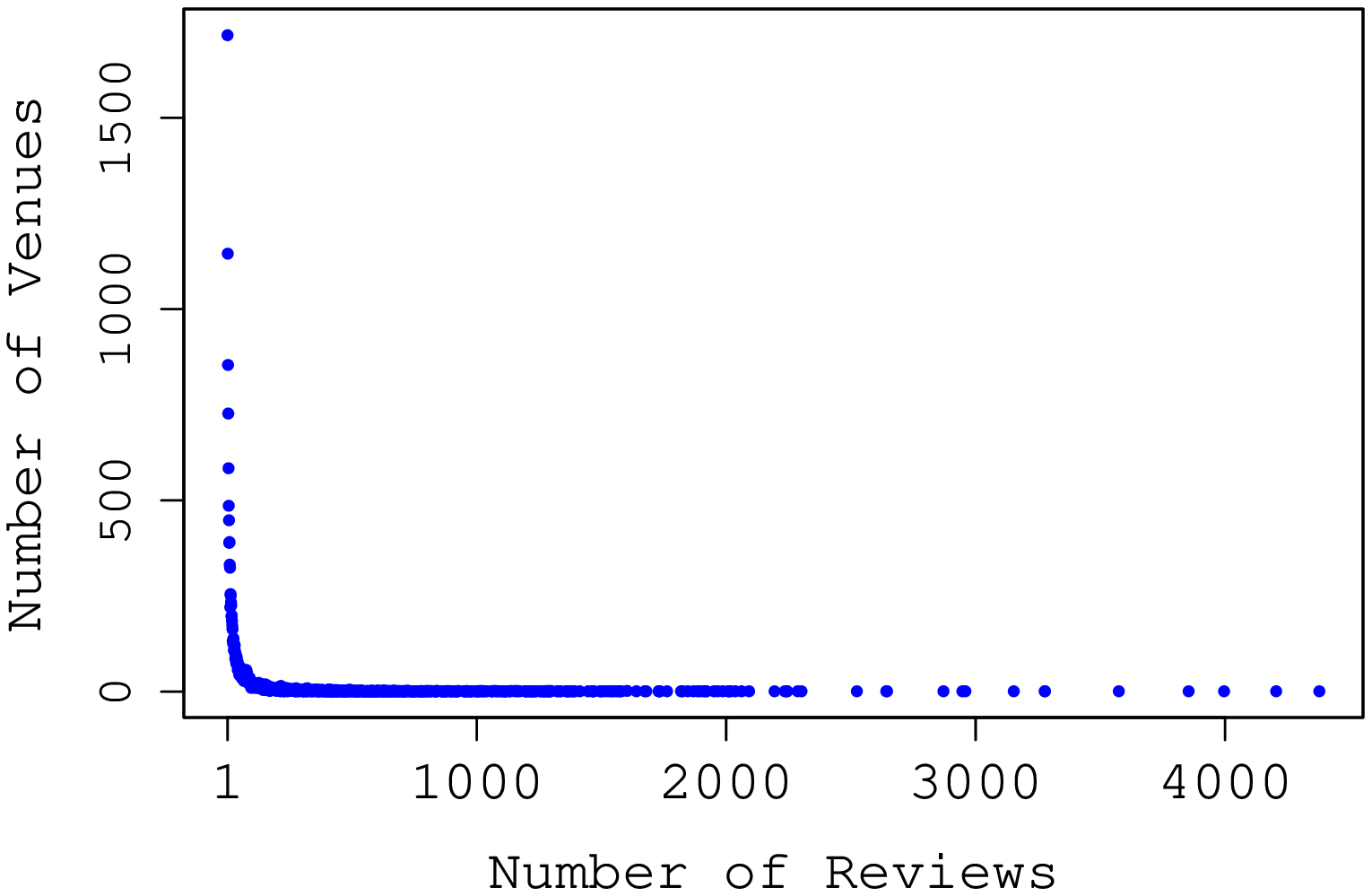}}}
\subfigure[]
{\label{fig:distDist}{\includegraphics[width=1.63in,height=1.35in]{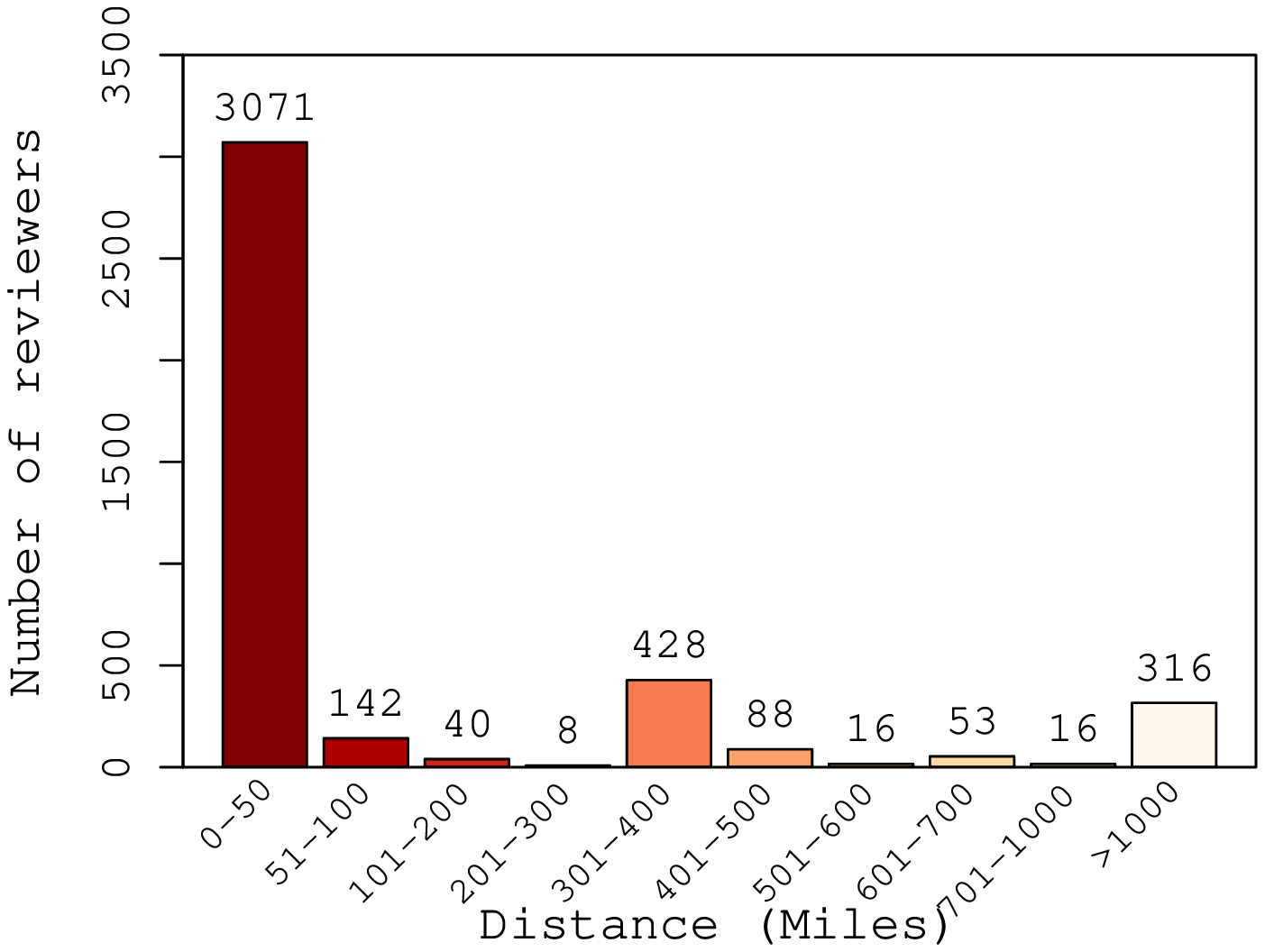}}}
\caption{Yelp venue stats:
(a) Distribution of the number of Yelp reviews per venue.
(b) Distribution of the distance from the venue ``Ike's Place'' to the home
cities of its reviewers.}
\end{figure}

We model the geosocial network (GSN) after Yelp~\cite{Yelp}. It consists of a
provider, $S$, hosting the system along with information about registered
venues, and serving a number of subscribers. To use the provider's services, a
client application needs to be downloaded and installed. Users register and
receive initial service credentials, including a unique user id.  We use the
terms $\emph{subscriber}$ and $\emph{user}$ interchangeably to refer to users
of the service and the term $\emph{client}$ to denote the software provided by
the service and installed by users on their devices.

The provider supports a set of businesses or venues, with an associated
geographic location (e.g., restaurants, yoga classes, towing companies, etc).
Users are encouraged to write reviews for visited locations, as well as
report their location, through \textit{check-ins} at venues where
they are present.

Participating venue owners need to install inexpensive equipment (e.g., a \$25
Raspberry PI~\cite{RaspberryPI}, a BeagleBoard or any Android smartphone).
%
Such equipment can also be used for other tasks including detecting fake user
check-ins~\cite{CP12} and preventing fake badges and incorrect rewards, and
validating social network (e.g., Yelp~\cite{Yelp}) reviews, thus eliminating
fake negative reviews. The advantages provided by such solutions can motivate
the small investment.

We have collected data from 16,199 venues throughout the U.S..  Besides the
name, location and type of venue, we have also collected all the reviews
provided for these venues, for a total of 1,096,044 reviews. For each review we
extracted the reviewer id, the date the review was written and the number of
check-ins performed. Moreover, we have collected data from 10,031 Yelp users,
including their id, location, number of friends and reviews, for a total of
646,017 reviews. Figure \ref{fig:distrev} shows the long-tail distribution of
the number of reviews per venue, for the collected venues.


\subsection{Location Centric Profiles}


Each user has a profile $P_U = \{u_1, u_2, .., u_d \}$, consisting of values
on $d$ dimensions (e.g., age, gender, home city, etc). Each dimension
has a range, or a set of possible values. Given a set of users $\mathcal{U}$
at location $L$, the \textit{location centric profile} at $L$, denoted by
$LCP(L)$ is the set $\{S_1, S_2, .., S_d \}$, where $S_i$ denotes the aggregate
statistics over the $i$-th dimension of profiles of users from $\mathcal{U}$.

In the following, we focus on a single profile dimension, $D$. We assume $D$
takes values over a range $R$ that can be discretized into a finite set of
sub-intervals (e.g., set of continuous disjoint intervals or discrete values).
Then, given an integer $b$, chosen to be dimension specific, we divide $R$ into
$b$ intervals/sets, $R_1, .., R_b$. For instance, gender maps naturally to
discrete values ($b=2$), while age can be divided into disjoint
sub-intervals, with a higher $b$ value.
We define the aggregate statistics $S$ for dimension $D$ of
$LCP(L)$ to consist of $b$ counters $c_1, .., c_k$; $c_i$ records the
number of users from $\mathcal{U}$ whose profile value on dimension $D$
falls within range $R_i$, $i=1..b$.

Figure \ref{fig:distDist} illustrates an LCP dimension: the distribution of the
(great-circle) distance in miles from a venue (``Ike's Place'' in San
Francisco, CA) and the home cities of its (4000+) reviewers.  Note that more
than 3000 reviews were left by locals, information that can be used by the
venue to better cater to its customers.



\subsection{Private LCP Requirements}
\label{sec:model:reqs}

\begin{figure}
\centering
\includegraphics[width=2.9in]{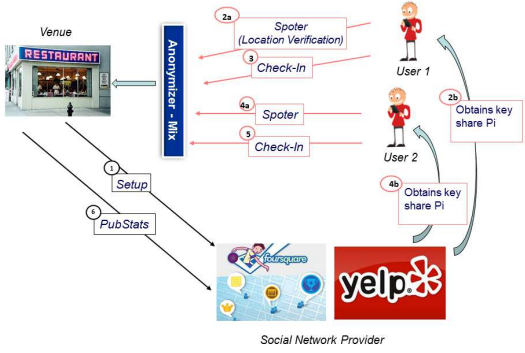}
\caption{Solution architecture ($k$=2). The red arrows denote anonymous
communication channels, whereas black arrows indicate authenticated
(and secure) communication channels.
\label{fig:solution}}
\end{figure}

We define a private LCP solution to be a set of functions, $PP(k)$ = $\{Setup$,
$Spoter$, $CheckIn$, $PubStats\}$, see Figure~\ref{fig:solution}. $Setup$ is
run by each venue where user statistics are collected, to generate parameters
for user check-ins.  To perform a check-in, a user first runs $Spoter$, to
prove her physical presence at the venue. $Spoter$ returns error if the
verification fails, success otherwise. If $Spoter$ is successful, $CheckIn$ is
run between the user and the venue, and allows the collection of profile
information from the user.  Specifically, if the user's profile value $v$ on
dimension $D$ falls within the range $R_i$, the counter $c_i$ is incremented by
1. Finally, $PubStats$ publishes collected LCPs.

Let $C_V$ be the set of counters defined at a venue $V$. Let $\bar{C}_V$ denote
the set of $b$ sets of counters derived from $C_V$, such that each set in
$\bar{C}_V$ has exactly one counter incremented over the set $C_V$. A private
LCP solution needs to satisfy the following properties:

\noindent
{\bf Location Correctness:}
Let $\mathcal{A}$ denote an adversary that controls the GSN provider and any
number of users. Let $\mathcal{C}$ be a challenger that controls a venue $V$.
$\mathcal{A}$ running as a user $U$ not present at $V$, has negligible
probability to successfully complete $Spoter$ at $V$.

\noindent
{\bf LCP Correctness:}
Let $\mathcal{A}$ denote an adversary that controls the GSN provider and any
number of users. Let $\mathcal{C}$ be a challenger that controls a venue $V$.
Let $C_V$ denote the set of counters at $V$ before $\mathcal{A}$ runs $CheckIn$
at $V$ and let $C'_V$ be the set of counters afterward. If $C'_V \notin
\bar{C}_V$, the $CheckIn$ completes successfully with only negligible
probability.

\noindent
{\bf $k$-Privacy:}
Let $\mathcal{A}$ denote an adversary that controls any number of venues and
let $\mathcal{C}$ denote a challenger controlling $k$ users.
$\mathcal{C}$ runs $Spoter$ followed by $CheckIn$ at a venue $V$ controlled by
$\mathcal{A}$ on behalf of $i < k$ users.  Let $C_i$ denote the resulting
counter set. For each $j=1..b$, $\mathcal{A}$ outputs $c'_j$, its guess of the
value of the $j$-th counter of $C_i$.  The advantage of $\mathcal{A}$,
$Adv(\mathcal{A}) = |Pr[C_i[j] = c'_j] - 1/(i+1)|$, defined for each $j=1..b$,
is negligible.

\noindent
{\bf Check-In Indistinguishability (CI-IND):}
Let a challenger $\mathcal{C}$ control two users $U_0$ and $U_1$ and let an
adversary $\mathcal{A}$ control any number of venues. $\mathcal{A}$ generates
randomly $q$ bits, $b_1, .., b_q$, and sends them to $\mathcal{C}$. For each
bit $b_i$, $i=1..q$, $\mathcal{C}$ runs $Spoter$ followed by $CheckIn$ on
behalf of user $U_{b_i}$. At the end of this step, $\mathcal{C}$ generates a random
bit $b$ and runs $Spoter$ followed by $CheckIn$ on behalf of $U_b$ at a venue
not used before.  $\mathcal{A}$ outputs a bit $b'$, its guess of $b$. The
advantage of $\mathcal{A}$, $Adv(\mathcal{A}) = |Pr[b' = b] - 1/2|$ is
negligible.


\subsection{Attacker Model}

We assume venue owners are malicious and will attempt to learn private
information from subscribers. Clients installed by users can be malicious,
attempting to bias LCPs constructed at target venues. We assume the GSN
provider does not collude with venues, but will try to learn private user
information.

\subsection{Cryptographic Tools}
\label{sec:model:tools}

\noindent
{\bf Homomorphic Cryptosystems.}
We use the Benaloh cryptosystem~\cite{B94}, an extension of the
Goldwasser-Micali~\cite{GM82}. It consists of three functions $(K, E, D)$, defined
as follows:


\begin{icompact}

\item
{\bf $K(k)$ - key generation}: $k$, an odd integer, is the size of the input
block. Select two large primes $p$ and $q$ such that $k | (p-1)$ and $gcd(k,
(p-1)/k) = 1$ and $gcd(k, q-1) = 1$. Let $n=pq$. Select $y \in \mathbb{Z}^*_n$,
such that $y^{(p-1)(q-1)/k}\ mod\ n \neq 1$. $n$ and $y$ are the public key and
$p$ and $q$ are the private key.

\item
{\bf $E(u,m)$}: Encrypt message $m \in \mathbb{Z}^*_k$, using a randomly chosen
value $u \in \mathbb{Z}^*_n$. Output $y^m u^k\ mod\ n$.

\item
{\bf $D(z)$}: Decrypt ciphertext $z$. Let $z = y^m u^k\ mod\ n$. If
$z^{(p-1)(q-1)/r} = 1$, then return $m=0$. Otherwise, for $i=1..k$, compute
$s_i = y^{-i} z\ mod\ n$. If $s_i = 1$, return $m=i$.

\end{icompact}


Benaloh's cryptosystem is additively homomorphic:\\ $E(u_1,m_1) E(u_2,m_2) =
E(u_1 u_2,m_1+m_2)$.  We further define the \textit{re-encryption} function
$RE(v, E(u,m))$ to be $y^m u^k v^k = E(u v, m)$. Note that the
re-encryption function can be invoked without knowledge of the message $m$.
Furthermore, it is possible to show that two ciphertexts are the encryption of
the same plaintext, without revealing the plaintext. That is, given $E(u,m)$
and $E(v,m)$, reveal $w = u^{-1} v$. Then, $E(v,m) = RE(w, E(u,m))$.

\noindent
{\bf Anonymizers.}
We use an anonymizer\cite{Chaum81,M98,PIK94,DMS04} that (i) operates correctly
-- the output corresponds to a permutation of the input and (ii) provides
privacy -- an observer is unable to determine which input element
corresponds to a given output element in any way better than guessing. In the
following we denote the anonymizer by $Mix$.

\noindent
{\bf Secret Sharing.}
Our constructions use a $(k,n)$ threshold secret sharing (TSS)~\cite{79shamir} solution. Given
a value $R$, TSS generates $n$ shares such that at least $k$ shares are needed
to reconstruct $R$. A $(k,n)$-TSS solution satisfies the property of
\textit{hiding}: An adversary (provided with access to a TSS oracle)
controlling the choice of two values $R_0$ and $R_1$ and given less than $k$
shares of $R_b$, $b \in_R \{0,1\}$, can guess the value of $b$ with
probability only negligible higher than 1/2.

\section{\profilr}
\label{sec:profilr}

Let \spotrv denote the device installed at venue $V$.  For each user profile
dimension $D$, \spotrv stores a set of \textit{encrypted counters} -- one for
each sub-range of $R$.

\paragraph*{\bf{Solution overview}}
Initially, and following each cycle of $k$ check-ins executed at venue $V$,
\spotrv initiates $Setup$, to request the provider $S$ to generate a new
Benaloh key pair. Thus, at each venue time is partitioned into \textit{cycles}:
a cycle completes once $k$ users have checked-in at the venue. The
communication during $Setup$ takes place over an authenticated and secure
channel (see Figure~\ref{fig:solution}).

When a user $U$ checks-in at venue $V$, it first engages in the $Spoter$
protocol with \spotrv. As shown in Figure~\ref{fig:solution}, this step
is performed over an anonymous channel, to preserve the user's (location)
privacy. $Spoter$ allows the venue to verify $U$'s physical presence
through a challenge/response protocol between \spotrv and the user device.
Furthermore, a successful run of $Spoter$ provides $U$ with a share of the
secret key employed in the Benaloh cryptosystem of the current cycle.
For each venue and user profile dimension, $S$ stores a set $Sh$ of shares
of the secret key that have been revealed so far.

Subsequently, $U$ runs $CheckIn$ with \spotrv, to send its share of the secret
key and to receive the encrypted counter sets. As shown in
Figure~\ref{fig:solution}, the communication takes place over an anonymous
channel to preserve $U$'s privacy. During $CheckIn$, for each dimension $D$,
$U$ increments the counter corresponding to her range, re-encrypts all counters
and sends the resulting set to \spotrv.  $U$ and \spotrv engage in a zero
knowledge protocol that allows \spotrv to verify $U$'s correct behavior:
exactly one counter has been incremented.  \spotrv stores the latest, proved to
be correct encrypted counter set, and inserts the secret key share into the set
$Sh$.

Once $k$ users successfully complete the $CheckIn$ procedure, marking the end
of a cycle, \spotrv runs $PubStats$ to reconstruct the private key, decrypt all
encrypted counters and publish the tally. The communication during $PubStats$
takes place over an authenticated channel (see Figure~\ref{fig:solution}).


\subsection{The Solution}

Let $C_i$ denote the set of encrypted counters at $V$, following the $i$-th
user run of $CheckIn$.  $C_i = \{ C_i[1],..,C_i[b] \}$, where $C_i[j]$ denotes
the encrypted counter corresponding to $R_j$, the $j$-th sub-range of $R$.  We
write $C_i[j] = E(u_j,u'_j,c_j,j)$ = $[E(u_j, c_j), E(u'_j, j)]$, where $u_j$
and $u'_j$ are random obfuscating factors and $E(u,m)$ denotes the Benaloh
encryption of message $m$ using random factor $u$.  That is, an encrypted
counter is stored for each sub-range of domain $R$ of dimension $D$. The
encrypted counter consists of two records, encoding the number of users whose
values on dimension $D$ fall within a particular sub-range of $R$.

Let $RE(v_j,v'_j,E(u_j,u'_j,c_j,j)$ denote the re-encryption of the $j$-th
record with two random values $v_j$ and $v'_j$:\\
$RE(v_j,v'_j,E(u_j,u'_j,c_j,j))$ = $[RE(v_j,E(u_j, c_j))$, $RE(v'_j,E(u'_j,
j))]$ = $[E(u_j v_j, c_j), E(u'_j v'_j, j)]$.  Let $C_i[j]++ =
E(u_j,u'_j,c_j+1,j)$ denote the encryption of the incremented $j$-th counter.
Note that incrementing the counter can be done without decrypting $C_i[j]$ or
knowing the current counter's value: $C_i[j]++$ = $[E(u_j, c_j) y, E(u'_j, j)]$
= $[y^{c_j+1} u_j^r, E(u'_j, j)]$ = $[E(u_j, c_j+1), E(u'_j, j)]$.

In the following we use the above definitions to introduce \profilr. \profilr
instantiates $PP(k)$, where $k$ is the privacy parameter.  The notation
$P(A(params_A), B(params_B))$ denotes the fact that protocol $P$ involves
participants $A$ and $B$, each with its own parameters.

\paragraph*{\bf{Setup}(V(),S($k$)):}
The provider $S$ runs the key generation function $K(k)$ of the Benaloh
cryptosystem (see Section~\ref{sec:model:tools}). Let $p$ and $q$ be the
private key and $n$ and $y$ the public key.  $S$ sends the public key to
\spotrv. \spotrv generates a signature key pair and registers the public key
with $S$. For each user profile dimension $D$ of range $R$, \spotrv performs
the following steps:


\begin{icompact}

\item
Initialize counters $c_1,..,c_b$ to 0. $b$ is the number of $R$'s sub-ranges.

\item
Generate $C_0 =\{E(x_1,x'_1,c_1,1),..,E(x_b,x'_b,c_b,b)\}$, where
$x_i,x'_i$, $i=1..b$ are randomly chosen values.
Store $C_0$ indexed on dimension $D$.

\item
Initialize the share set $S_{key} = \emptyset$.

\end{icompact}


\paragraph*{\bf{Spoter}(U(K),V(),S($k$)):}
$U$ sets up an anonymous connection with \spotrv, e.g., by using
fresh, random MAC and IP address values. \spotrv initiates a challenge/response
protocol, by sending to $U$ the currently sampled time $T$, an expiration
interval $\Delta T$ and a fresh random value $R$. The user's device generates a
hash of these values and sends the result back to \spotrv.  \spotrv ensures
that the response is received within a specific interval from the challenge (see
Section~\ref{sec:eval} for values and discussion). If
the verification succeeds, \spotrv uses its private key to sign a timestamped
token and sends the result to $U$. $U$ contacts $S$ over $Mix$ and sends the
token signed by \spotrv. $S$ verifies $V$'s signature as well as the freshness
(and single use) of the token. Let $U$ be the $i$-th user checking-in at $V$.
If the verifications pass and $i \le k$, $S$ uses the $(k,n)$ TSS to compute a
share of $p$ (Benaloh secret key, factor of the modulus $n$). Let $p_i$ be the
share of $p$. $S$ sends the (signed) share $p_i$ to $U$. If $i > k$, $S$ calls
$Setup$ to generate new parameters for $V$. 

\paragraph*{\bf{CheckIn}(U($p_i$, n, V), V(n, y, $C_{i-1}$, $S_{key}$))}:
$U$ uses the same random MAC and IP addresses as in the previous $Spoter$ run.
Executes only if the previous run of $Spoter$ is successful. Let $U$ be the
$i$-th user checking-in at $V$. Then, $C_{i-1}$ is the current set of encrypted
counters. \spotrv sends $C_{i-1}$ to $U$. Let $v$, $U$'s value on dimension
$D$, be within $R$'s $j$-th sub-range, i.e., $v \in R_j$. $U$ runs the following
steps:


\begin{icompact}

\item
Generate $b$ pairs of random values $\{ (v_1,v'_1),..,(v_b,v'_b) \}$.  Compute
the new encrypted counter set $C_i$, where the order of the counters in $C_i$
is identical to $C_{i-1}$: $C_i$ =\\ $\{ RE(v_l,v'_l,C_{i-1}[l])| l=1..b,l \neq
j\}$ $\cup$ $RE(v_j,v'_j,C_{i-1}[j]++)$.

\item
Send $C_i$ along with the signed (by $S$) share $p_i$ of the private key $p$ to
$V$.

\end{icompact}


\noindent
If \spotrv successfully verifies the signature of $S$ on the share $p_i$, $U$
and \spotrv engage in a zero knowledge protocol ZK-CTR (see
Section~\ref{sec:profilr:zkp}).  ZK-CTR allows $U$ to prove that $C_i$
is a correct re-encryption of $C_{i-1}$: only one counter of $C_{i-1}$ has
been incremented. If the proof verifies, \spotrv replaces $C_{i-1}$ with $C_i$
and ads the share $p_i$ to the set $S_{key}$.

\paragraph*{\bf{PubStats}(V($C_k$,Sh,V),S(p,q))}:
\spotrv performs the following actions:


\begin{icompact}

\item
If $|Sh| < k$, abort.

\item
If $|Sh| = k$, use the $k$ shares to reconstruct $p$, the private Benaloh key.

\item
Use $p$ and $q = n/p$ to decrypt each record in $C_k$, the final set of
counters at $V$. Publish results.

\end{icompact}

\vspace{5pt}

\subsection{ZK-CTR: Proof of Correctness} 
\label{sec:profilr:zkp}


We now present the zero knowledge proof of the set $C_i$ being a correct
re-encryption of the set $C_{i-1}$, i.e., a single counter has been incremented.
Let ZK-CTR(i) denote the protocol run for sets $C_{i-1}$ and $C_i$. $U$ and
\spotrv run the following steps $s$ times:


\begin{icompact}

\item
$U$ generates random values $(t_1,t'_1),..,(t_b,t'_b)$ and random permutation
$\pi$, then sends to \spotrv the proof set
$P_{i-1} = \pi \{ RE(t_l,t'_l,C_{i-1}[l]), l = 1..b \}$.

\item
$U$ generates random values $(w_1,w'_1),..,(w_b,w'_b)$, then sends to \spotrv the
proof set $P_i = \pi \{ RE(w_l,w'_l,C_i[l]), l = 1..b \}$

\item
\spotrv generates a random bit $a$ and sends it to $U$.

\item
If $a=0$, $U$ reveals random values $(t_1,t'_1),..,(t_b,t'_b)$ and
$(w_1,w'_1),..,(w_b,w'_b)$.  \spotrv verifies that for each $l=1..b$,
$RE(t_l,t'_l,C_{i-1}[l])$ occurs in $P_{i-1}$ exactly once, and that for each
$l=1..b$, $RE(w_l,w'_l,C_i[l])$ occurs in $P_i$ exactly once.

\item
If $a=1$, $U$ reveals $o_l = v_l w_l t^{-1}_l$ and $o'_l = v'_l w'_l
t'^{-1}_l$, for all $l=1..b$ along with $j$, the position in $P_{i-1}$ and
$P_i$ of the incremented counter. \spotrv verifies that for all $l=1..b, l \neq
j$, $RE(o_l, o'_l, P_{i-1}[l]) = P_i[l]$ and $RE(o_j, o'_j, P_{i-1}[j] y)
= P_i[j]$.

\item
If any verification fails, \spotrv aborts the protocol.

\end{icompact}


\subsection{Preventing Illegal Votes}


For simplicity of presentation, we have avoided the Sybil attack problem:
participants that cheat through multiple accounts they control or by exploiting
the anonymizer. For instance, a rogue venue owner, controlling $k$-1 Sybil user
accounts or simulating $k$-1 check-ins, can use \profilr to reveal the profile
of a real user. Conversely, a rogue user (including the venue) could bias the
statistics built by the venue (and even deny service) by checking-in multiple
times in a short interval.

Sybil detection techniques (see Section~\ref{sec:related}) can be used to
control the number of fake, Sybil accounts.  However, the use
of the anonymizer prevents the provider and the use of the unique IP
and MAC addresses prevents the venue from differentiating
between interactions with the same or different accounts. In this section we
propose a solution, that when used in conjunction with Sybil detection tools,
mitigates this problem. The solution introduces a trade-off between privacy and
security.

Specifically, we divide time into epochs (e.g., one day long). A user can
check-in at any venue at most once per epoch.  When  active, once per epoch
$e$, each user $U$ contacts the provider $S$ over an authenticated channel. $U$
and $S$ run a blind signature~\cite{blind} protocol: $U$ obtains the signature of $S$ on a
random value, $R_{U,e}$. $S$ does not sign more than one value for $U$ for any
epoch. In runs of $Spoter$ and $CheckIn$ during epoch $e$, $U$ uses $R_{U,e}$
as its pseudonym (i.e., MAC and IP address). Venues can verify the validity of
the pseudonym using $S$'s signature.  A venue accepts a single $CheckIn$ per
epoch from any pseudonym, thus limiting the user's impact on the LCP.



\subsection{Analysis}


Given a set of encrypted counters $C$, let $\bar{\mathcal{C}}$ denote the set
of re-encryptions of records of $C$, where only one record has its counter
incremented. We introduce the following theorem.

\begin{thm}
ZK-CTR(i) is a ZK proof of $C_i \in \bar{\mathcal{C}}_{i-1}$.
\label{thm:zk}
\end{thm}

\begin{proof}
We need to prove completeness, soundness and zero-knowledge. For completeness,
if $C_i \in \bar{\mathcal{C}}_{i-1}$, in each of the $s$ steps, $U$ succeeds
to convince $S$, irrespective of the challenge bit $a$. If $a=0$, $U$ can produce
the random obfuscating values proving that the proof sets $P_{i-1}$ and $P_i$
are correctly generated from $C_{i-1}$ and $C_i$. If $a=1$, $U$ can build
the obfuscating factors proving that $P_i \in \bar{\mathcal{P}}_{i-1}$.

For soundness, we need to prove that if $C_i \notin \bar{\mathcal{C}}_{i-1}$,
$U$ cannot convince $S$ unless with negligible probability. For simplicity
reasons, we assume $C_i \notin \bar{\mathcal{C}}_{i-1}$ due to a single record
in $C_i$ being ``bad'': $C_{i-1}[j] = E(u_j,u'_j,c_j,j)$ and $C_i[j] =
E(v_j,v'_j,c'_j,j')$. In any round of the ZK-CTR protocol, $U$ has two options
for cheating.  First, $U$ could count on the bit $a$ to come up 0. Then, $U$
builds $P_{i-1}[j] = E(u_j t_j, u'_j t'_j, c_j, j)$ and $P_i[j] = E(v_j w_j,
v'_j w'_j, c'_j, j')$. If however $a=1$, $U$ has to come up with a value
$\alpha_j$, such that $RE(\alpha_j, E(u_j,c_j) = E(v'_j,c'_j)$ or $RE(\alpha_j,
E(u_j,c_j+1) = E(v'_j,c'_j)$. In the first case, this means $y^{c_j} (u_j
\alpha_j)^k = y_{c'_j} v'^k_j\ mod\ n$. Without knowing $n$'s factorization,
$U$ cannot compute $k$'s inverse modulo $\phi(n)$. Then, the equation is
satisfied only if $c'_j = c_j + z k$, for an integer $z$. Note however that
Benaloh's cryptosystem only works for values in $\mathbb{Z}^*_k$, making this
condition impossible to satisfy. The second case is similar. The second
cheating option is to assume $a$ will be 1 and build $P_i[j]$ to be a
re-encryption of $P_{i-1}[j]$.  It is then straightforward to see that if
$a=0$, $U$ can only succeed in convincing $S$, if $c'_j = c_j + z k$, which
we have shown is impossible for $z \neq 0$. Thus, in each round, $U$ can
only cheat with probability 1/2. Following $s$ rounds, this probability
becomes $1/2^s$.

We now show that ZK-CTR conveys no knowledge to any verifier, even one that
deviates arbitrarily from the protocol.  We prove this by following the
approach from~\cite{GMR89,GMW91}.  Specifically, let $S^*$ be an arbitrary,
fixed, expected polynomial time ITM. We generate an expected polynomial time
machine $M^*$ that, without being given access to the client, produces an
output whose probability distribution is identical to the probability
distribution of the output of $<C,S^*>$.

We now build $M^*$ that uses $S^*$ as a black box many times. Whenever $M^*$
invokes $S^*$, it places input $x=(L_0,L_1)$ on its input tape $IT_S$ and a
fixed sequence of random bits on its random tape, $RT_S$. The input $x$
consists of $L_0 = C_0$ and $L_1 = C_1$.  The content of the input
communication tape for $S^*$, $CT_S$ will consist of tuples
$(P_{2i},P_{2i+1},\pi_i)$, where $P_{2i}$ and $P_{2i+1}$ are sets and $\pi_i$
is a permutation.  The output of $M^*$ consists of two tapes: the random-record
tape $RT_M$ and the communication-record tape $CT_M$. $RT_M$ contains the
prefix of the random bit string $r$ read by $S^*$. The machine $M^*$ works as
follows (round $i$):

\noindent
{\bf Step 1}
$M^*$ chooses a random bit $a \in_R \{0,1\}$. If $a=0$, $M^*$ picks a
random permutation $\pi_i$, generates $t_l,t'_l$, $l=1..b$ randomly and
computes $P_{2i} = \pi_i \{ RE(t_l,t'_l,C_{i-1}[l]), l = 1..b \}$.
It then generates random values $w_l,w'_l$, $l=1..b$, randomly and computes the
set $P_{2i+1} = \pi_i \{ RE(w_l,w'_l,C_i[l]), l = 1..b \}$. Note that $M^*$
does not need to know the counters to perform this operation.
If $a=1$, $M^*$ generates a random set $P_{2i}$, then generates
random values $o_l,o'_l$ randomly, $l=1..b$. It then generates
a random $j \in 1..b$ and computes $P_{2i+1}$ such that for all $l=1..b, l \neq j$,
$RE(o_l, o'_l, P_{2i}[l]) = P_{2i+1}[l]$ and for the $j$-th position,
$RE(o_j, o'_j, P_{2i}[j] y) = P_{2i+1}[j]$.

\noindent
{\bf Step 2}
$M^*$ sets\\
$b=S^*(x,r;P_0,P_1,\pi_0,..,P_{2i-2},P_{2i-1},\pi_{i-1},P_{2i},P_{2i+1})$. That
is, $b$ is the output of $S^*$ on input $x$ and random string $r$ after
receiving $i-1$ pairs $P_{2j},P_{2j+1},\pi_j)$, $j=1..i-1$ and proof
$P_{2i},P_{2i+1}$ on its communication tape $CT_S$.  We have the following
three cases.

(Case 1). $a=b=0$. $M^*$ can produce $t_l,t'_l,w_l,w'_l$, $l=1..b$ and $\pi_i$ to prove
that $P_{2i} = \pi_i \{ RE(t_l,t'_l,C_{i-1}[l]), l = 1..b \}$ and $P_{2i+1} =
\pi_i \{ RE(w_l,w'_l,C_i[l]), l = 1..b \}$.  $M^*$ sets $b_i$ to $b$, appends
the tuple $(P_{2i},P_{2i+1},\pi_i,b_i)$ to $CT_M$ and proceeds to the next
round (i+1).

(Case 2). $a=b=1$. $M^*$ can produce $o_l,o'_l$, $l=1..b$, and index j such
that $RE(o_l, o'_l, P_{2i}[l]) = P_{2i+1}[l]$, $l=1..b, l \neq j$ and $RE(o_j,
o'_j, P_{2i}[j] y) = P_{2i+1}[j]$. $M^*$ sets $b_i$ to $b$, appends the tuple
$(P_{2i},P_{2i+1},\pi_i,b_i)$ to $CT_M$ and proceeds to the next round (i+1).

(Case 3). $a \neq b$. $M^*$ discards all the values of the current
iteration and repeats the current round (Step 1 and 2).

If all rounds are completed, $M^*$ halts and outputs $(x,r',CT_M)$, where
$r'$ is the prefix of the random bits $r$ scanned by $S^*$ on input $x$.
We first prove that $M^*$ terminates in expected polynomial time and then
that the output distribution of $M^*$ is the same as the
output distribution of $S*$ when interacting with the client, on input $(L_0,L_1)$.

\begin{lemma}
$M^*$ terminates in expected polynomial time.
\end{lemma}

\begin{proof}
Given $C_0$ and $C_1$, during the $i$-th round $P_{2i}$ and $P_{2i+1}$ are
either built from $C_0$ and $C_1$ or from each other. During each run of round
$i$, the bit $a$ is chosen independently. Then $P_{2i}$ and $P_{2i+1}$ are also
chosen independently. This implies that the probability that $a=b$ is 1/2 and
the expected number of repetitions of round $i$ is 2. $S^*$ is expected
polynomial time, which implies that $M^*$ is also polynomial time.
\end{proof}

\begin{lemma}
The probability distribution of $<C,S^*>(L_0,L-1)>$ and of $M^*(L_0,L_1)$
are identical.
\end{lemma}

\begin{proof}
The output of $<C,S^*>(L_0,L_1)>$ and of $M^*(L_0,L_1)$ consists of a sequence
of $t$ tuples of format\\ $(P_{2i},P_{2i+1},\pi_i,b_i)$. Let
$\Pi_{M^*}^{(x,r,i)}$ and $\Pi_{CS^*}^{(x,r,i)}$ be the probability
distributions of the first $i$ tuples output by $M^*$ and $<C,S^*>$. We need
to show that for any fixed random input $r$,
$\Pi_{M^*}^{(x,r,t)}=\Pi_{CS^*}^{(x,r,t)}$. We prove this by induction.  The
base case, where $i=0$, holds immediately. In the induction step we assume that
$\Pi_{M^*}^{(x,r,i)}=\Pi_{CS^*}^{(x,r,i)}=T^{(i)}$. We need to prove that the
$i+1$st tuples in $\Pi_{M^*}^{(x,r,i+1)}$, denoted by $\Pi_{M^*}^{(i+1)}$ and
in $\Pi_{CS^*}^{(x,r,i+1)}$, denoted by $\Pi_{CS^*}^{(i+1)}$ have the same
distribution.
We show that $\Pi_{M^*}^{(i+1)}$ and $\Pi_{CS^*}^{(i+1)}$ are uniform over
the set
$V = \{ (P_{2i},P_{2i+1},\pi_i,b) | b=S^*(x,r,T^{(i)}||P) \wedge ((P_{2i}=\pi_i{RE(C_0)},
P_{2i+1}=\pi_i{RE(C_1)}$,
$if\ b=0)$ $\vee$ $(P_{2i+1}[l]=RE(P_{2i}[l]), l=1..b, l\neq j, P_{2i+1}[j]=y RE(P_{2i}[j]),\ if\ b=1) \}$.
For $\Pi_{CS^*}^{(i+1)}$, this is the case, by construction.
If $\Pi_{M^*}^{(i+1)}$ has output, it is also uniformly
distributed in $V$.
\end{proof}

$M^*$ terminates in expected polynomial time and its output has
the same distribution as the output of the interaction between $S^*$ and a
client. Thus, the theorem follows.
\end{proof}

We can now prove the following results.

\begin{thm}
\profilr is $k$-private.
\end{thm}

\begin{proof}(Sketch)
Following the definition from Section~\ref{sec:model:reqs}, let us assume that
the adversary $\mathcal{A}$ has access to an encrypted counter set $C_{i}$
generated after $\mathcal{C}$ has run $Spoter$ followed by $CheckIn$ on behalf
of $i < k$ different users. The records of set $C_i$ are encrypted and
$\mathcal{A}$ has $i$ shares of the private key. For any $j=1..b$, let $c'_j$
be $\mathcal{A}$'s guess of the value of the $j$-th counter in $C_i$.  If
$|Pr[C_i[j] = c'_j] - 1/(k+1)| = \epsilon$ is non-negligible we can use
$\mathcal{A}$ to construct an adversary $\mathcal{B}$ that has $\epsilon$
advantage in the (i) semantic security game of Benaloh or in the (ii) hiding
game of the $(k,n)$ TSS. We start with the first reduction.  $\mathcal{B}$
generates two messages $M_0 = 0$ and $M_1 = 1$ and sends them to the challenger
$\mathcal{C}$. $\mathcal{C}$ picks a bit $d \in_R \{0,1\}$ and sends to $\mathcal{B}$ the
value $E(u, M_d)$, where $u$ is random and $E$ denotes Benaloh's encryption
function. $\mathcal{B}$ initiates a new game with $\mathcal{A}$, with counters
set to 0. $\mathcal{B}$ runs $Spoter$ and $CheckIn$ (acting as challenger) with
$\mathcal{A}$. $\mathcal{B}$ re-encrypts all counters from $\mathcal{A}$,
except the $j$-th one, which it replaces with $E(u, M_d)$. $\mathcal{B}$ runs
ZK-CTR with $\mathcal{A}$ (used as a black box) a polynomial number of times
until it succeeds. $\mathcal{A}$ outputs its guess of the values of all
counters. $\mathcal{B}$ sends the guess for the $j$-th counter to
$\mathcal{C}$. The advantage of $\mathcal{B}$ in this game comes entirely from
the advantage provided by $\mathcal{A}$.

For the second reduction, $\mathcal{B}$ runs $Setup$ as the provider and
obtains the secret key $p_0$ and $p_1$ (renamed from $p$ and $q$).
$\mathcal{B}$ sends $p_0$ and $p_1$ to the challenger $\mathcal{C}$, as its
choice of two random values. $\mathcal{C}$ generates a random bit $a$, uses the
$(k,n)$ TSS to generate $i < k$ shares of $p_a$, $sh_1, .., sh_i$, and sends
them to $\mathcal{B}$.  $\mathcal{B}$ generates a new random prime $q$ and
picks randomly a bit $d$.  Let the Benaloh modulus be $n = p_d q$. Then, acting
as $i$ different users, $U_j$, $j=1..i$ $\mathcal{B}$ runs $Spoter$ with $S$
(which it also controls) to obtain $S$'s signature on $sh_j$. For each of the
$i$ users, $\mathcal{B}$ runs $CheckIn$ with $\mathcal{A}$. At the end of the
process, $\mathcal{A}$ outputs its guess of the encrypted counters. If the
guess is correct on more than $d/(j+1)$ counters, $\mathcal{B}$ sends $d$ to
$\mathcal{C}$ as its guess for $a$. Otherwise, it sends $\bar{d}$. Thus,
$\mathcal{B}$'s advantage in the hiding game of TSS is equivalent to
$\mathcal{A}$'s advantage against \profilr.
\end{proof}

\begin{thm}
\label{thm:loc}
\profilr ensures location correctness.
\end{thm}

\begin{proof}
The user's location is verified in the $Spoter$ protocol. A single malicious
user, not present at venue $V$, is unable to establish a connection with the
device deployed at $V$, \spotrv. Thus, the user is unable to participate in the
challenge/response protocol and receive at its completion a provider signed
share of the Benaloh secret key. Without the share, the user is unable to
initiate the $CheckIn$ protocol.  Two (or more) attackers can launch wormhole
attacks: one attacker present at $V$, acts a a proxy and relays information
between \spotrv and a remote attacker. This may allow the remote attacker to
successfully run $Spoter$ and $CheckIn$ at $V$.  In Section~\ref{sec:eval} we
present experimental proof that $Spoter$ detects wormhole attacks.
\end{proof}

\begin{thm}
\profilr is LCP correct.
\end{thm}

\begin{proof}(Sketch)
A user $U$ can alter the LCP of a venue $V$ in two ways.  First,
during the ZK-CTR protocol, it modifies more than one counter or corrupts (at
least ) one counter. The soundness property of ZK-CTR, proved in
Theorem~\ref{thm:zk} shows this attack succeeds with probability $1/2^s$.
Second, it attempts to prevent $V$ from decrypting the counter sets after $k$
users have run CheckIn. This can be done by preventing \spotrv from
reconstructing the private Benaloh key. Key shares are however signed by the
provider, allowing \spotrv to detect invalid shares.
\end{proof}

\begin{thm}
\profilr provides CI-IND.
\end{thm}

\begin{proof}(Sketch)
Let $\mathcal{A}$ be an adversary that has an $\epsilon$ advantage in the
CI-IND game. We assume the challenger does not run $Spoter$ and $CheckIn$ twice
for the same (user, epoch) pair -- otherwise the use of the signed pseudonyms
provides an advantage to $\mathcal{A}$. Note that if pseudonyms are not used,
this requirement is not necessary. Moreover, no identifying information is
sent by users during $Spoter$ and $CheckIn$: the pseudonyms are \textit{blindly}
signed by $S$, and all communication with $S$ takes places over $Mix$.
\end{proof}

\section{Snapshot LCP}
\label{sec:realtime}


We extend \profilr to allow not only venues but also users to collect
\textit{snapshot} LCPs of other, co-located users.  To achieve this, we take
advantage of the ability of most modern mobile devices (e.g., smartphones,
tablets) to setup ad hoc networks.  Devices establish local connections with
neighboring devices and privately compute the instantaneous aggregate LCP of
their profiles.

\subsection{Snapshot \profilr}

We assume a user $U$ co-located with $k$ other users $U_1,..,U_k$. $U$ needs to
generate the LCP of their profiles, without infrastructure, GSN provider or
venue support. An additional difficulty then, is that participating users need
assurances that their profiles will not be revealed to $U$. However, one
advantage of this setup is that location verification is not needed: $U$
intrinsically determines co-location with $U_1,..,U_k$. Snapshot \profilr
consists of three protocols, $\{Setup, LCPGen, PubStats\}$:

\paragraph*{\bf{Setup}($U(r),U_1,..,U_k$())}:
$U$ performs the following steps:


\begin{icompact}

\item
Run the key generation function $K(r)$ of the Benaloh cryptosystem (see
Section~\ref{sec:model:tools}). Send the public key $n$ and $y$ to each user
$U_1,..,U_k$.

\item
Engage in a multi-party secure function evaluation protocol~\cite{B00,JJ00}
with $U_1,..,U_k$ to generate shares of a public value $R < n$. At the end of
the protocol, each user $U_i$ has a share $R_i$, such that $R_1 .. R_k = R\
mod\ n$ and $R_i$ is only known to $U_i$.

\item
Assign each of the $k$ users a unique label between 1 and $k$. Let
$U_1,..,U_k$ denote this order.

\item
Generate $C_0=\{E(x_1,x'_1,0,1),..,E(x_b,x'_b,0,b)\}$, where $x_i,x'_i$,
$i=1..b$ are randomly chosen. Store $C_0$ indexed on dimension $D$.

\end{icompact}


\noindent
Each of the $k$ users engages in a 1-on-1 $LCPGen$ with $U$ to
privately and correctly contribute her profile to $U$'s LCP.

\paragraph*{\bf{LCPGen}($U(C_{i-1}),U_i()$)}:
Let $C_{i-1}$ be the encrypted counters after $U_1,..,U_{i-1}$ have completed
the protocol with $U$. $U$ sends $C_{i-1}$ to $U_i$. $U_i$ runs the following:


\begin{icompact}

\item
Generate random values $(v_1,v'_1),..,(v_b,v'_b)$.  Let $j$ be the index of the
range where $U_i$ fits on dimension $D$.

\item
Compute the new encrypted counter set $C_i$ as:
$C_i = \{ RE(v_l,v'_l,C_{i-1}[l]) R_i\ mod\ n | l=1..b,l \neq j \}$ $\cup$
$RE(v_j,v'_j,C_{i-1}[j]++) R_i\ mod\ n \}$ and send it to $U$.

\item
Engage in a ZK-CTR protocol to prove that $C_i \in \bar{C}_{i-1}$. The
only modification to the ZK-CTR protocol is that all re-encrypted values are
also multiplied with $R_i\ mod\ n$, $U_i$'s share of the public value $R$. If
the proof verifies, $U$ replaces $C_{i-1}$ with $C_i$.

\end{icompact}


\noindent
After completing $LCPGen$ with $U_1,..,U_k$, $U$'s encrypted
counter set is $C_k = \{E_j = E(u_j,u'_j,c_j,j) R_1 .. R_k | j = 1 ..d\}$,
where $u_j$ and $u_j'$ are the product of the obfuscation factors used by
$U_1,..,U_k$ in their re-encryptions. The following protocol enables $U$
to retrieve the snapshot LCP.

\paragraph*{\bf{PubStats}($U(C_k$))}:
Compute $E_j K$, $\forall j=1..d$, where $K = R^{-1}\ mod\ n$ ($R = R_1 ..
R_k$), decrypt the outcome using the private key ($p$, $q$) and publish the
resulting counter value.

\noindent
Even though $U$ has the private key allowing it to decrypt any Benaloh
ciphertext, the use of the secret $R_i$ values prevents it from
learning the profile of $U_i$, $i=1..k$.



\section{iSafe: Context Aware Safety}
\label{sec:isafe}

\begin{figure}
\centering
\includegraphics[width=2.3in]{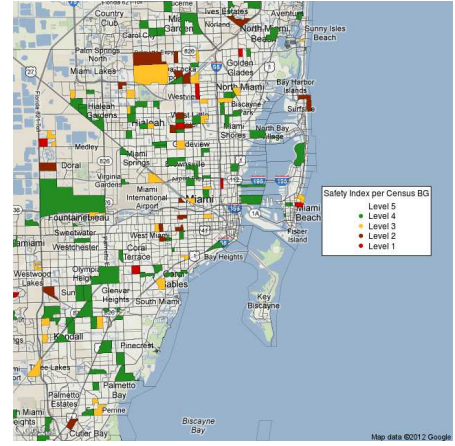}
\caption{Static crime indexes computed over crimes reported during 2010 in the
Miami-Dade county.
\label{fig:crime:index}}
\end{figure}

We introduce iSafe, an application built on \profilr.  iSafe uses the context
of users, in terms of their location, time, other people present, to build a
\textit{safety} representation.  Quantifying the safety of a user based on her
current context can be further used to provide safe walking directions and
context-aware smartphone authentication protocols (i.e., more complex
authentication protocols in unsafe locations).
iSafe combines information collected from Yelp with Census~\cite{censusdata}
and historical crime databases as well as context collected by the users'
mobile devices. We have access to the Miami-Dade county~\cite{crimeterrafly}
area crime and Census datasets since 2007. Each record in the crime dataset is
labeled with a crime type (e.g., homicide, larceny, robbery) as well as the
geographic location and time of occurrence.

iSafe assigns static safety labels to Census-defined geographic blocks.  While
beyond the scope here, we note that the safety index is inversely proportional
to the weighted average of the crimes committed in the block.
Figure~\ref{fig:crime:index} shows the color-coded safety index for each block
group in the Miami-Dade county (FL) in 2010.
%
iSafe uses the static block safety indexes to compute safety labels of mobile
users. The safety label of a user is an average over the safety indexes of the
blocks visited by the user. Blocks visited more frequently, have an inherently
higher impact on the user's safety label. Block and user safety labels take
values in the $[0,1]$ interval; 1 is the safest label.

iSafe uses \profilr to privately compute the safety labels for Yelp venues: the
distribution of safety indexes of users that reviewed them. To achieve this,
iSafe divides the $[0,1]$ safety range into a discrete set of disjoint
sub-intervals, and assigns a counter to each sub-interval. Each venue privately
retrieves the distribution of the safety values of its reviewers (the counters
of users fitting the corresponding sub-intervals). Finally, the safety index of
the venue is the weighted average of the aggregated counts. The normalized
weights are either the upper bound value or the middle point of their
corresponding sub-intervals.

Besides this venue-centric approach, iSafe also uses snapshot \profilr to
privately aggregate the safety labels of co-located user devices and
distributively obtain the real-time image of the safety of their location.

\subsection{Implementation}
\label{sec:implementation}

We implemented iSafe as a (i) web server, (ii) a browser
plugin running in the user's browser and (iii) a mobile application.  We use
Apache Tomcat 6.0.35 to route requests (exposed to the client through a REST
API interface) to our server-side component. The server-side component relies
on the latest servlet v3.0 which offers additional features including
asynchronous support, making the server-side processing much more efficient. We
implemented the browser plugin for the Chrome browser using HTML, CSS and
Javascript. The plugin interacts with Yelp pages and the web server, using
content scripts (Chrome specific components that let us access the browser's
native API) and cross-origin XMLHttpRequests.

\begin{figure}
\centering
\includegraphics[width=3.3in]{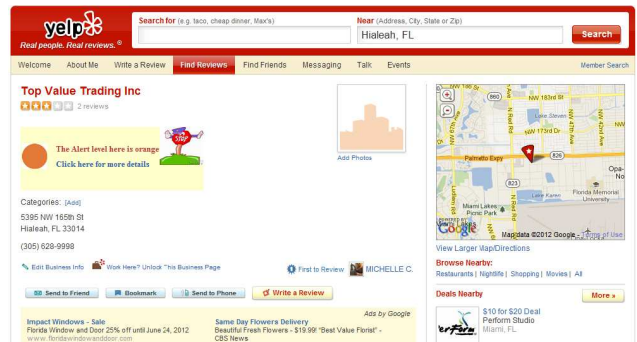}
\caption{Snapshot of iSafe's plugin functionality for a Yelp venue.
The orange circle indicates the venue's safety level.
\label{fig:isnapshot}}
\end{figure}

The browser plugin becomes active when the user navigates to a Yelp page.
For user and venue pages, the plugin parses their HTML file and retrieves their
reviews. We employ a stateful approach, where the server's DB stores all
reviews of pages previously accessed by users. This enables significant time
savings, as the plugin needs to send to the web server only reviews written
after the date of the last user's access to the page.
Given the venue's set of reviews, the server determines the corresponding
reviewers. Since we do not have access to the location history of users, to
compute a user's security label we rely on the venues reviewed by the user: The
user safety is computed as an average over the safety labels of the blocks
containing the venues reviewed by the user.
Given the safety labels of reviewers, we run \profilr to determine their
distribution and identify the safety level of the venue.  The server sends back
the safety level of the venue, which the plugin displays in the browser.
Figure~\ref{fig:isnapshot} shows iSafe's extension to the Yelp page of
the venue ``Top Value Trading Inc.'' in Hialeah, FL (central left yellow
rectangle containing iSafe's safety recommendations).

We have also implemented an Android front-end for iSafe's snapshot LCPs.  We
used the standard Java security library to implement the cryptographic
primitives employed by \profilr.  For secret sharing, we used Shamir's scheme
and for digital signatures we used RSA. We also used the kSOAP2 library to
enable SOAP functionality on the Android app. Figure~\ref{fig:isafe} shows a
snapshot of the iSafe Android app on a Samsung Admire smartphone. We used the
Google map API to facilitate the location based service employed by our
approach.

\begin{figure}
\centering
\subfigure[]
{\label{fig:movil1}{\includegraphics[width=1.4in,height=2.3in]{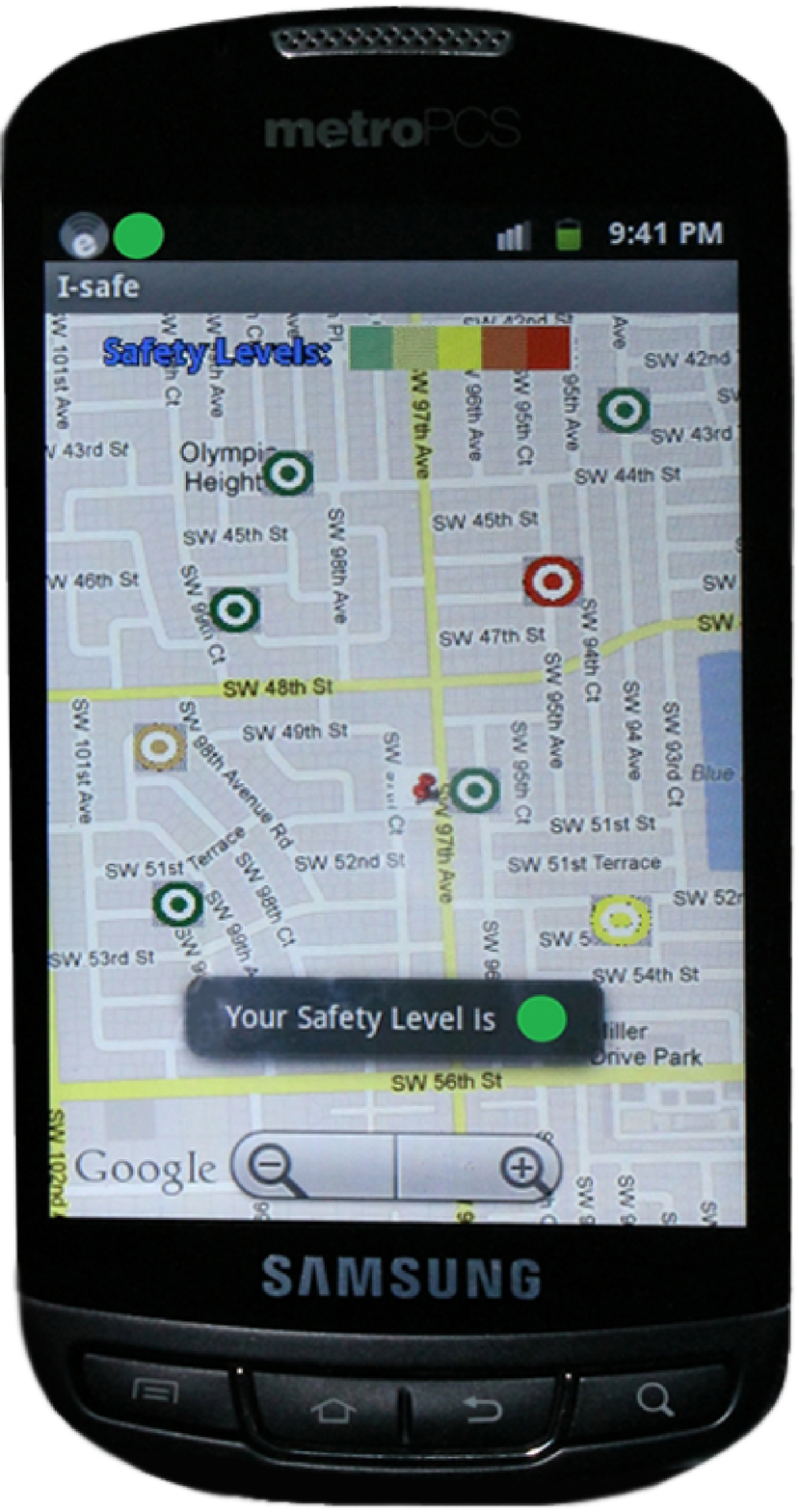}}}
\subfigure[]
{\label{fig:movil2}{\includegraphics[width=1.5in,height=2.3in]{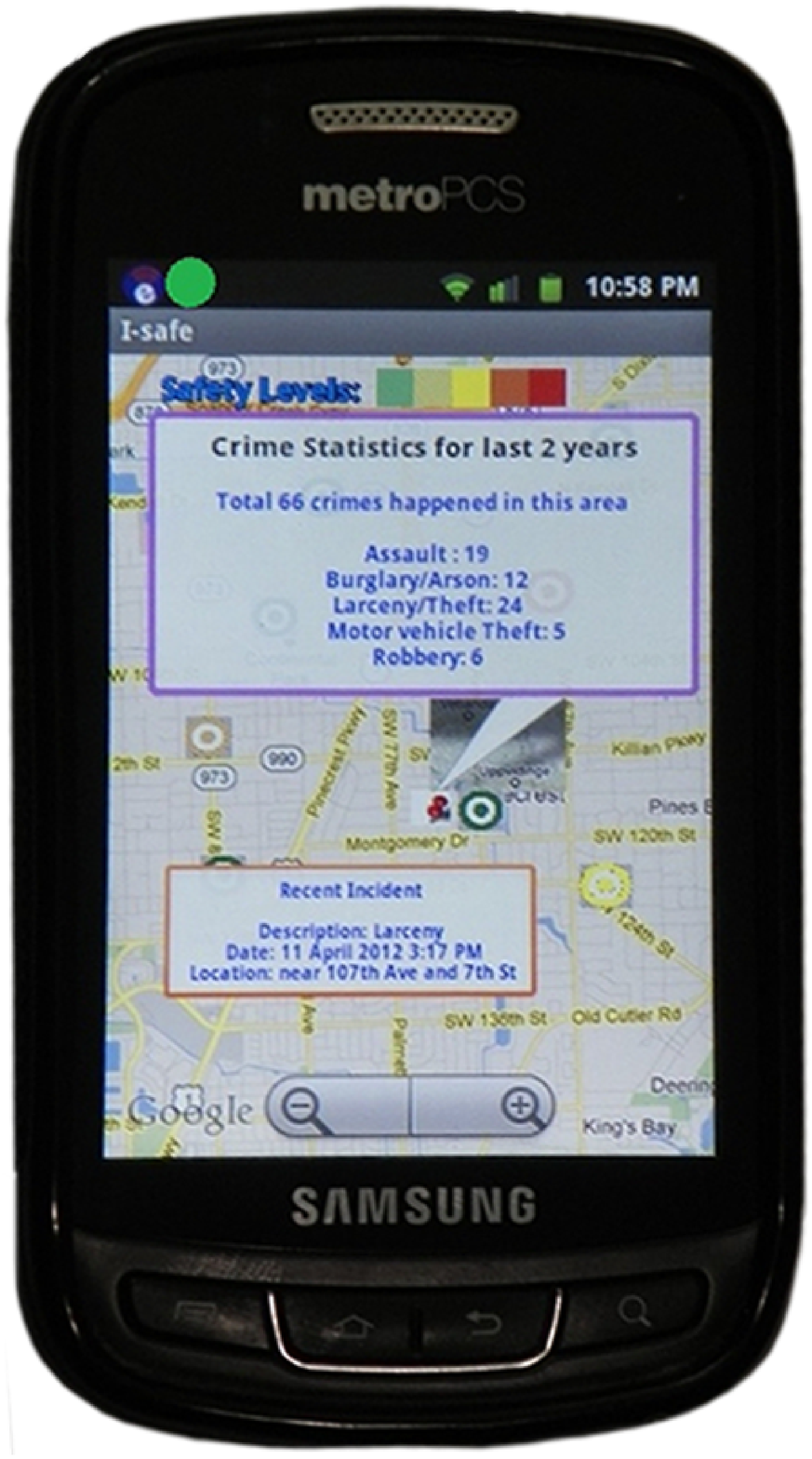}}}
\caption{Snapshots of iSafe on Android.}
\label{fig:isafe}
\end{figure}


\section{Evaluation}
\label{sec:eval}

For testing purposes we have used Samsung Admire smartphones running Android OS
Gingerbread 2.3 with a 800MHz CPU and a Dell laptop equipped with a 2.4GHz
Intel Core i5 processor and 4GB of RAM for the server. For local connectivity
the devices used their 802.11b/g Wi-Fi interfaces. All reported values are
averages taken over at least 10 independent protocol runs.

\noindent
{\bf iSafe:}
Figure~\ref{fig:isafe:collect} shows the overhead of the iSafe plugin when
collecting the reviews of a venue browsed by the user, as a function of the
number of reviews the venue has. It includes the cost to request each review
page, parse and process the data for transfer. The experiments were performed
on the Dell laptop. It exhibits a sub-linear dependence on the number of
reviews of the venue (under 1s for 10 reviews but under 30s for 4000 reviews),
showing that Yelp's delay for successive requests decreases. While even for 500
reviews the overhead is less than 5s, we note that this cost is incurred only
once per venue.  Subsequent accesses to the same venue, by any other user will
no longer incur this overhead.

\begin{figure*}
\centering
\subfigure[]
{\label{fig:isafe:collect}{\includegraphics[width=2.1in]{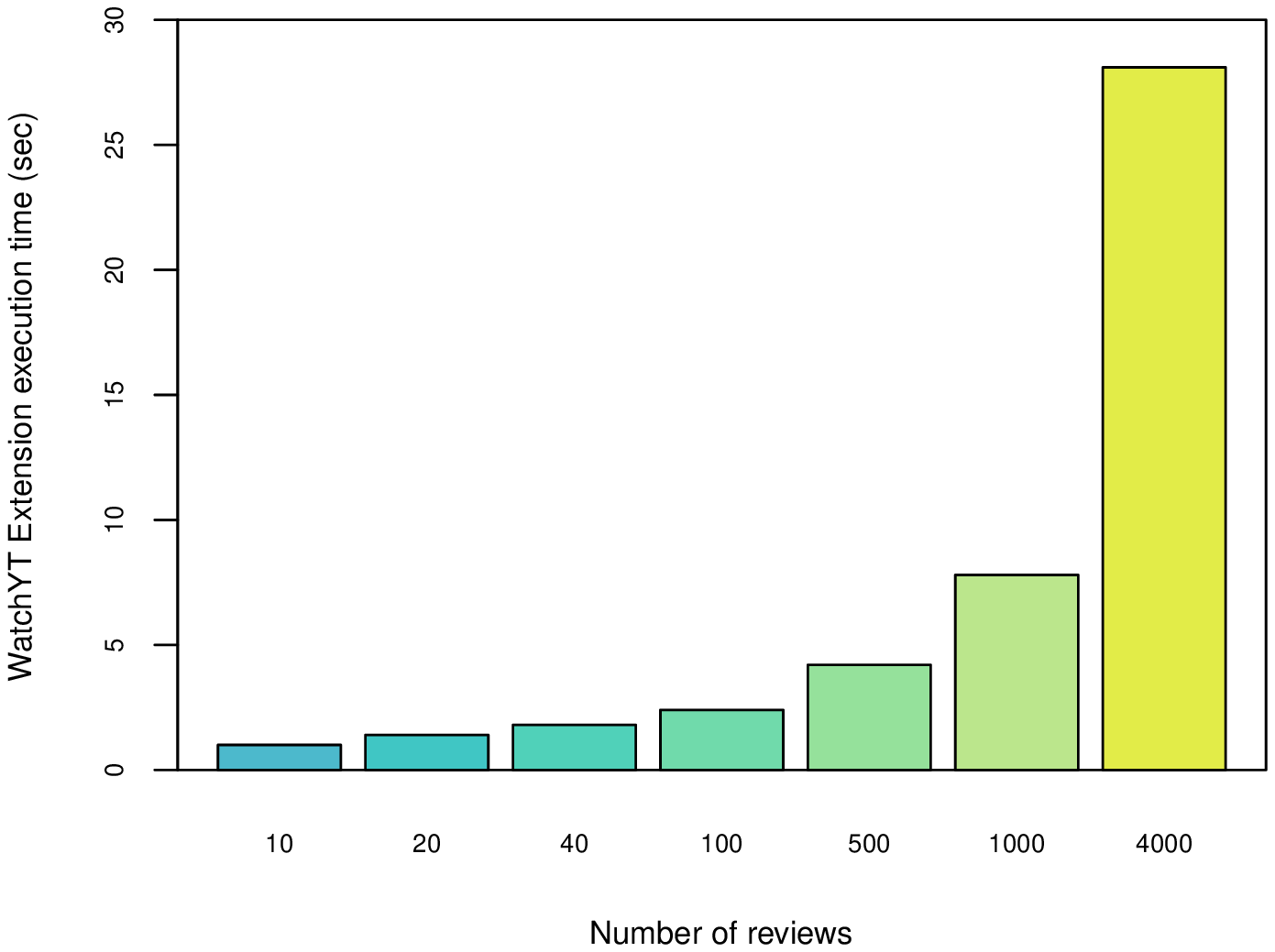}}}
\subfigure[]
{\label{fig:setup}{\includegraphics[width=2.1in]{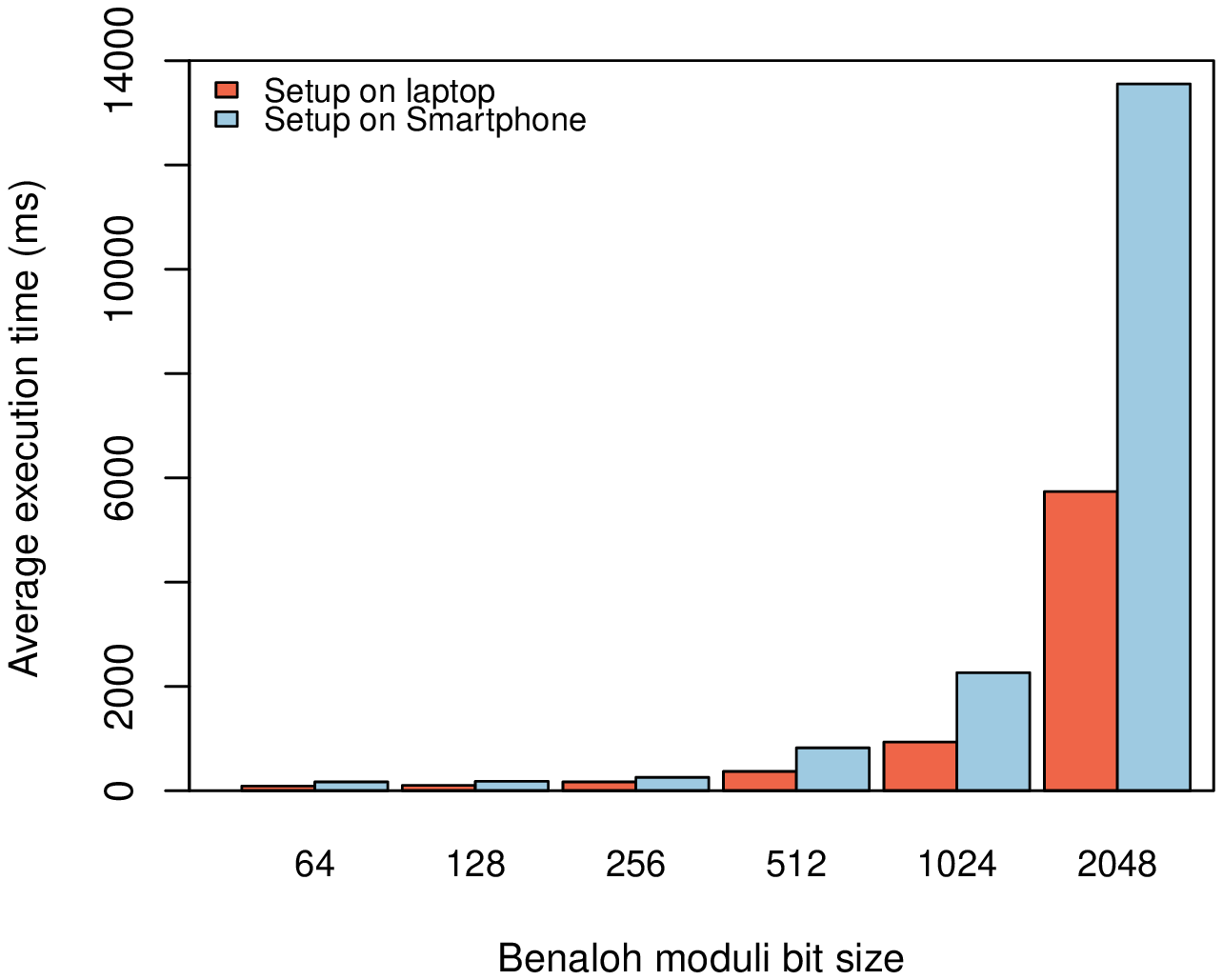}}}
\subfigure[]
{\label{fig:storage}{\includegraphics[width=1.9in]{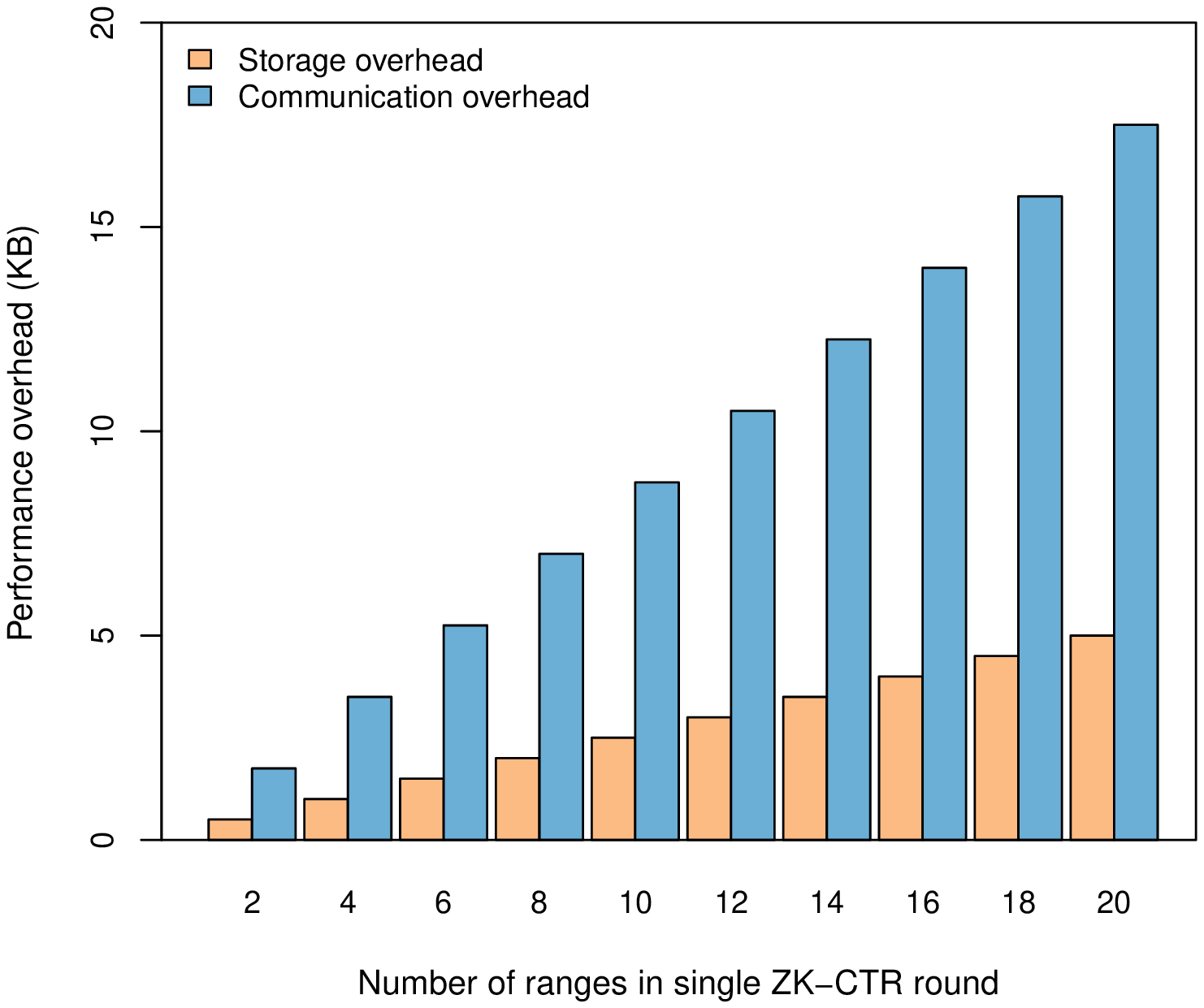}}}
\caption{(a) iSafe browser plugin overhead: Collecting reviews from venues, as a
function of the number of reviews.
(b) $Setup$ dependence on Benaloh modulus size.
(c) Storage and communication overhead (in KB) as a function of range count.}
\end{figure*}

\noindent
{\bf $Spoter$'s wormhole defenses:}
Wormhole attacks are best detected through timing analysis.  We have tested
Spoter using a smartphone connected over ad hoc Wi-Fi to the laptop. The
round-trip Wi-Fi latency is under 3ms.  On the Android device, the time
required to compute a (SHA-512) hash is 0.6ms.  The overhead imposed by
$Spoter$ on a wormhole attack is the Wi-Fi round-trip latency, plus the hash
time (0.003ms on the laptop operations), plus the wired round-trip
communication latency. The one-way communication overhead between the two
attackers, if performed over the wired network, is at least 19ms (we tested
with systems in Miami, San Francisco and Chicago). In total, $Spoter$ imposes
an overhead on a wormhole attack (43ms) that is almost 12 times the overhead
imposed on an honest user (3.6ms). Thus, wormhole attacks are easily detectable
in $Spoter$.


\subsection{\profilr Evaluation}


We have first measured the overhead of the $Setup$ operation.  We set the
number of ranges of the domain $D$ to be 5,  Shamir's TSS group size to 1024 bits
and RSA's modulus size to 1024 bits.  Figure~\ref{fig:setup} shows the $Setup$
overhead on the smartphone and laptop platforms, when the Benaloh modulus size
ranges from 64 to 2048 bits. Note that even a resource constrained smartphone
takes only 2.2s for 1024 bit sizes (0.9s on a laptop). A marked increase can
be noticed for the smartphone when the Benaloh bit size is 2048 bit long - 13.5s.
We note however that this cost is amortized over multiple check-in runs.

\begin{figure}
\centering
\subfigure[]
{\label{fig:zk:bitsize}{\includegraphics[width=1.63in,height=1.35in]{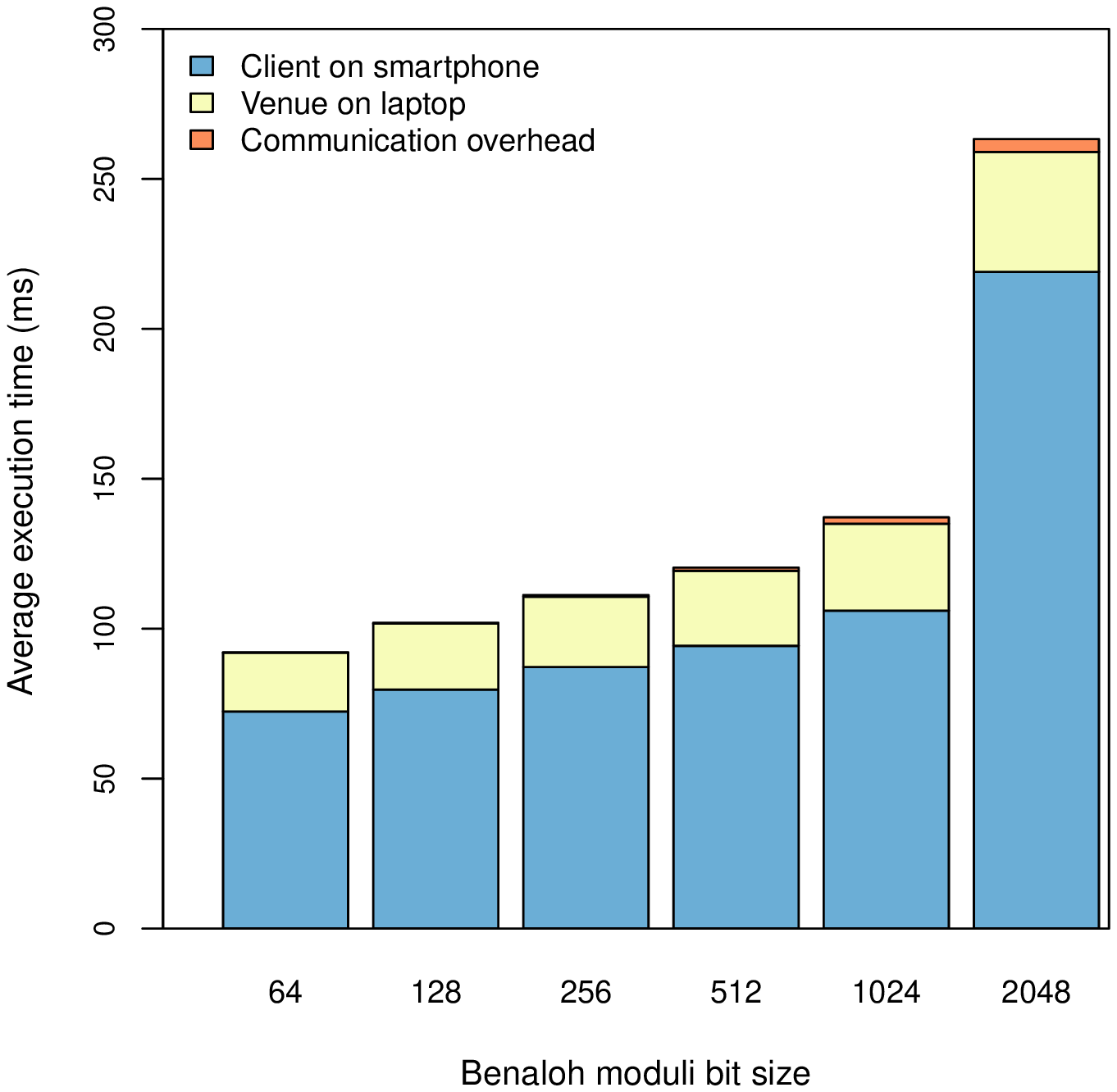}}}
\subfigure[]
{\label{fig:zk:rounds}{\includegraphics[width=1.63in,height=1.35in]{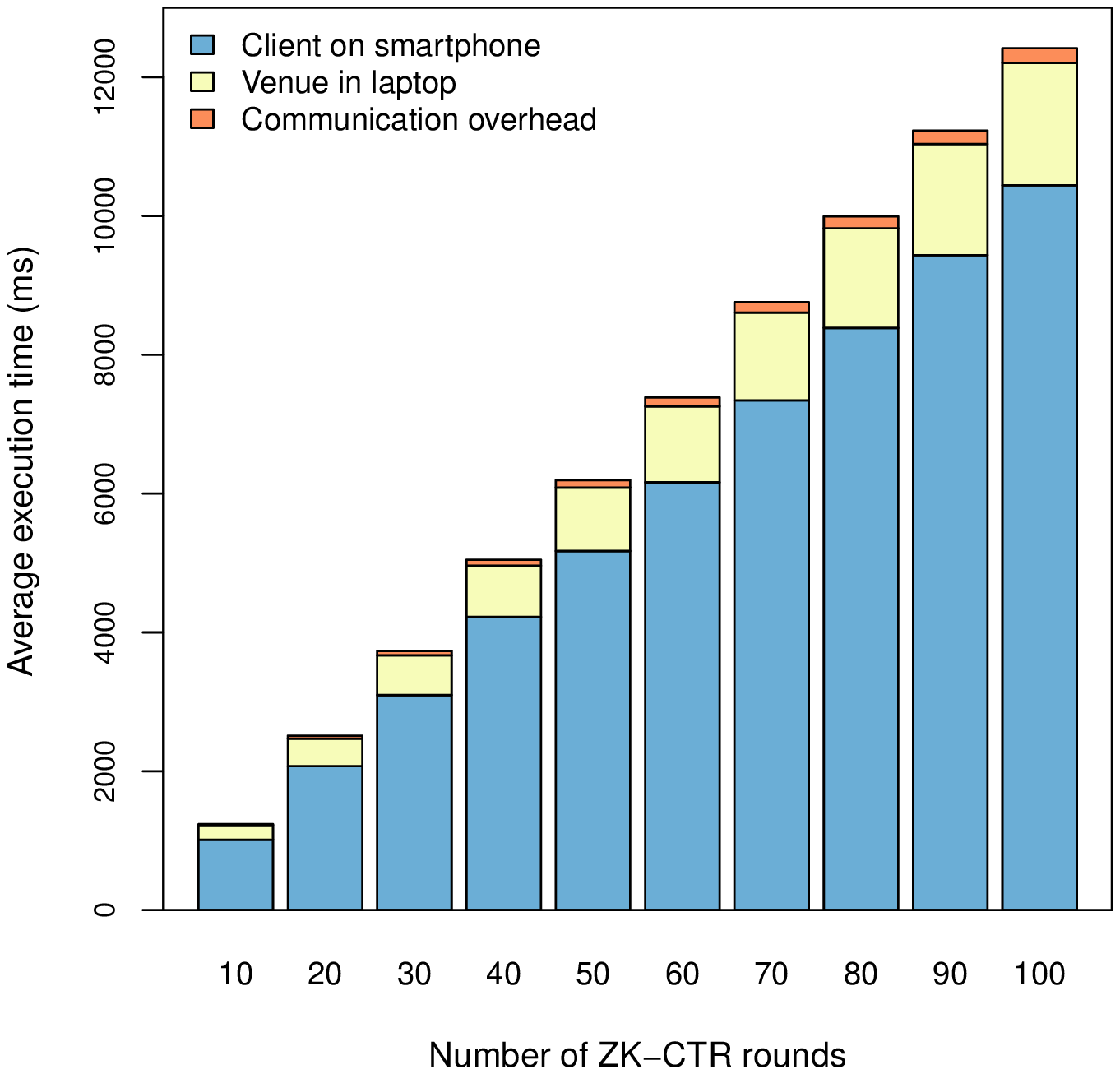}}}
\caption{ZK-CTR Performance:
(a) Dependence on the Benaloh modulus size.
(b) Dependence on the number of proof rounds.}
\end{figure}


We now focus on the most resource consuming component of \profilr: the ZK-CTR
protocol. We measure the client and venue (\spotrv) computation overhead as
well as their communication overhead. We set the number of sub-ranges of domain
$D$ to 5. We tested the client side running on the smartphone and the venue
component executing on the laptop. Figure~\ref{fig:zk:bitsize} shows the
dependence of the three costs for a single round of ZK-CTR on the Benaloh
modulus size. Given the more efficient venue component and the superior
computation capabilities of the laptop, the venue component has a
much smaller overhead. The communication overhead is the smallest, exhibiting a
linear increase with bit size. For a Benaloh key size of 1024 bits, the average
end-to-end overhead of a single ZK-CTR round is 135ms. The venue component is
29ms and the client component is 106ms. Furthermore,
Figure~\ref{fig:zk:rounds} shows the overheads of these components as a
function of the number of ZK-CTR rounds, when the Benaloh key size is 1024 bit
long. For 30 rounds, when a cheating client's probability of success is $2^{-30}$,
the total overhead is 3.6s.



We further examine the communication overhead in terms of bits transferred
during ZK-CTR between a client and a venue. Let $N$ be the Benaloh modulus size
and $B$ the sub-range count of domain $D$. The communication overhead in a
single ZK-CTR round is $4BN + 3BN = 7BN$. The second component of the sum is
due to the average outcome of the challenge bit. Figure~\ref{fig:storage} shows
the dependency of the communication overhead (in KB) on $B$, when $N=1024$.
Even when $B=20$, the communication overhead is around 17KB.
Figure~\ref{fig:storage} shows also the storage overhead (at a venue). The
storage overhead is only a fraction of the (single round) communication
overhead, $2BN$. For a single dimension, with 20 sub-ranges, the overhead is
5KB.


\section{Related Work}
\label{sec:related}


Golle et al.~\cite{GMM06} proposed techniques allowing pollsters to collect
user data while ensuring the privacy of the users. The privacy is proved at
``runtime'': if the pollster leaks private data, it will be exposed
probabilistically. Our work also allow entities to collect private user data,
however, the collectors are never allowed direct access to private user data.

Toubiana et. al~\cite{Adnostic} proposed Adnostic, a privacy preserving ad
targeting architecture. Users have a profile that allows the private matching
of relevant ads.  While \profilr can be used to privately provide location
centric targeted ads, its main goal is different - to compute location
(venue) centric profiles that preserve the privacy of contributing users.

Manweiler et al.~\cite{MSC09} proposed SMILE, a privacy-preserving
``missed-connections'' service similar to Craigslist, where the service provider
is untrusted and users do not have existing relationships. The solution is
distributed, allowing users to anonymously prove to each other the existence of
a past encounter. While we have a similar setup, our work addresses a different
problem, of privately collecting location centric user profile aggregates.

Location and temporal cloaking techniques, or introducing errors in reported
locations in order to provide 1-out-of-k anonymity have been initially proposed
in~\cite{GG03}, followed by a significant body of
work~\cite{HGHBWHBAJ08,OTGH10,PMX09,GDSB09}.  We note that \profilr provides an
orthogonal notion of $k$-anonymity: instead of reporting intervals containing
$k$ other users, we allow the construction of location centric profiles only
when $k$ users have reported their location. Computed LCPs hide the profiles
the users: user profiles are anonymous, only aggregates are available for
inspection, and interactions with venues and the provider are
indistinguishable.

Our work relies on the assumption that participants cannot control a large
number of fake, Sybil accounts. One way to ensure this property is to use
existing Sybil detection techniques. Danezis and Mittal~\cite{DM09} proposed a
centralized SybilInfer solution based in Bayesian inference.  Yu et al.
proposed distributed solutions, SybilGuard~\cite{YKGF06} and
SybilLimit~\cite{YGKX10}, that use online social networks to protect
peer-to-peer network against Sybil nodes. They rely on the fast mixing property
of social networks and the limited connectivity of Sybil nodes to non-Sybil
nodes.

Significant work has been done recently to preserve the privacy of users from
the online social network provider. Cutillo et al.~\cite{CMS09} proposed
Safebook, a distributed online social networks where insiders are protected
from external observers through the inherent flow of information in the system.
Tootoonchian et al.~\cite{TSGW09} proposed Lockr, a system for improving the
privacy of social networks. It achieves this by using the concept of a social
attestation, which is a credential proving a social relationship.  Baden et
al.~\cite{BSB09} introduced Persona, a distributed social network with
distributed account data storage. Sun et al.~\cite{SZF10} proposed a similar
solution, extended with revocation capabilities through the use of broadcast
encryption.
While we rely on distributed online social networks, our goal is to protect the
privacy of users while also allowing venues to collect certain user statistics.


\section{Conclusions}
\label{sec:conclusions}


We have proposed (i) novel mechanisms for building aggregate location-centric
profiles while maintaining the privacy of participating users and ensuring
their honesty during the process and (ii) centralized and distributed,
real-time variants of the solution, along with applications that can benefit
from the construction of such profiles. We have shown that our solutions are
efficient, even when executed on resource constrained mobile devices.


%

\bibliographystyle{unsrt}
\bibliography{anonymous,bogdan,crime,crypto,ecash,foursquare,key,location,osn,sybil,sfe,privacy,zk}

\begin{thebibliography}{10}

\bibitem{Yelp}
Yelp.
\newblock \url{http://www.yelp.com}.

\bibitem{foursquare}
Foursquare.
\newblock \url{https://foursquare.com/}.

\bibitem{FBPlaces}
{Facebook Places}.
\newblock \url{http://www.facebook.com/places}.

\bibitem{KW10}
Balachander Krishnamurthy and Craig~E. Wills.
\newblock On the leakage of personally identifiable information via online
  social networks.
\newblock {\em Computer Communication Review}, 40(1):112--117, 2010.

\bibitem{FBPrivacy}
Emily Steel and Geoffrey Fowler.
\newblock Facebook in privacy breach.
\newblock
  \url{http://online.wsj.com/article/SB10001424052702304772804575558484075236968.html}.

\bibitem{cheat1}
Foursquare~Official Blog.
\newblock On foursquare, cheating, and claiming mayorships from your couch.
\newblock http://goo.gl/F1Yn5, 2011.

\bibitem{LocationSpoofer}
Big Boss.
\newblock Location spoofer.
\newblock http://goo.gl/59HMk, 2011.

\bibitem{GPSCheat}
Gpscheat!
\newblock \url{http://www.gpscheat.com/}.

\bibitem{RaspberryPI}
{Raspberry Pi. An ARM GNU/Linux box for \$25. Take a byte!}
\newblock \url{http://www.raspberrypi.org/}.

\bibitem{CP12}
Bogdan Carbunar and Rahul Potharaju.
\newblock {You unlocked the Mt. Everest Badge on Foursquare! Countering
  Location Fraud in GeoSocial Networks}.
\newblock In {\em {To appear in Proceedings of the 9th IEEE International
  Conference on Mobile Ad hoc and Sensor Systems (MASS)}}, 2012.

\bibitem{B94}
Josh Benaloh.
\newblock Dense probabilistic encryption.
\newblock In {\em Proceedings of the Workshop on Selected Areas of
  Cryptography}, pages 120--128, 1994.

\bibitem{GM82}
Shafi Goldwasser and Silvio Micali.
\newblock Probabilistic encryption \& how to play mental poker keeping secret
  all partial information.
\newblock In {\em Proceedings of the fourteenth annual ACM symposium on Theory
  of computing}, STOC '82, pages 365--377, New York, NY, USA, 1982. ACM.

\bibitem{Chaum81}
David~L. Chaum.
\newblock Untraceable electronic mail, return addresses, and digital
  pseudonyms.
\newblock {\em Commun. ACM}, 24(2), 1981.

\bibitem{M98}
Masayuki Abe.
\newblock Universally verifiable mix-net with verification work indendent of
  the number of mix-servers.
\newblock In {\em Proceedings of EUROCRYPT}, pages 437--447, 1998.

\bibitem{PIK94}
Choonsik Park, Kazutomo Itoh, and Kaoru Kurosawa.
\newblock Efficient anonymous channel and all/nothing election scheme.
\newblock In {\em EUROCRYPT '93: Workshop on the theory and application of
  cryptographic techniques on Advances in cryptology}, pages 248--259, 1994.

\bibitem{DMS04}
Roger Dingledine, Nick Mathewson, and Paul~F. Syverson.
\newblock Tor: The second-generation onion router.
\newblock In {\em USENIX Security Symposium}, pages 303--320, 2004.

\bibitem{79shamir}
Adi Shamir.
\newblock How to share a secret.
\newblock {\em Communications of the ACM}, 22(11):612--613, 1979.

\bibitem{blind}
David Chaum.
\newblock Blind signatures for untraceable payments.
\newblock In {\em Advances in Cryptology: Proceedings of CRYPTO '82}, pages
  199--203, 1982.

\bibitem{GMR89}
S.~Goldwasser, S.~Micali, and C.~Rackoff.
\newblock The knowledge complexity of interactive proof systems.
\newblock {\em SIAM J. Comput.}, 18(1), 1989.

\bibitem{GMW91}
Oded Goldreich, Silvio Micali, and Avi Wigderson.
\newblock Proofs that yield nothing but their validity or all languages in np
  have zero-knowledge proof systems.
\newblock {\em J. ACM}, 38(3), 1991.

\bibitem{B00}
Donald Beaver.
\newblock Minimal-latency secure function evaluation.
\newblock In {\em Proceedings of the 19th international conference on Theory
  and application of cryptographic techniques}, EUROCRYPT'00, pages 335--350,
  Berlin, Heidelberg, 2000. Springer-Verlag.

\bibitem{JJ00}
Markus Jakobsson and Ari Juels.
\newblock Mix and match: Secure function evaluation via ciphertexts.
\newblock In {\em Advances in Cryptology - ASIACRYPT 2000, 6th International
  Conference on the Theory and Application of Cryptology and Information
  Security}, pages 162--177, 2000.

\bibitem{censusdata}
United~States Census.
\newblock 2010 census.
\newblock http://2010.census.gov/2010census/, 2010.

\bibitem{crimeterrafly}
{Terrafly Project}.
\newblock {Crimes and Incidents Reported by Miami-Dade County and Municipal
  Police Departments}.
\newblock
  \url{http://vn4.cs.fiu.edu/cgi-bin/arquery.cgi?lat=25.81&long=-80.12&category=crime_dade}.

\bibitem{GMM06}
Philippe Golle, Frank McSherry, and Ilya Mironov.
\newblock Data collection with self-enforcing privacy.
\newblock In Rebecca Wright, Sabrina De~Capitani di~Vimercati, and Vitaly
  Shmatikov, editors, {\em ACM Conference on Computer and Communications
  Security---ACM CCS 2006}, pages 69--78. ACM, October 2006.

\bibitem{Adnostic}
Vincent Toubiana, Arvind Narayanan, Dan Boneh, Helen Nissenbaum, and Solon
  Barocas.
\newblock Adnostic: Privacy preserving targeted advertising.
\newblock In {\em NDSS}, 2010.

\bibitem{MSC09}
Justin Manweiler, Ryan Scudellari, and Landon~P. Cox.
\newblock Smile: encounter-based trust for mobile social services.
\newblock In {\em Proceedings of the 16th ACM conference on Computer and
  communications security}, CCS '09, pages 246--255, New York, NY, USA, 2009.
  ACM.

\bibitem{GG03}
Marco Gruteser and Dirk Grunwald.
\newblock Anonymous usage of location-based services through spatial and
  temporal cloaking.
\newblock In {\em Proceedings of MobiSys}, 2003.

\bibitem{HGHBWHBAJ08}
Baik Hoh, Marco Gruteser, Ryan Herring, Jeff Ban, Dan Work, Juan-Carlos
  Herrera, Re~Bayen, Murali Annavaram, and Quinn Jacobson.
\newblock {Virtual Trip Lines for Distributed Privacy-Preserving Traffic
  Monitoring}.
\newblock In {\em Proceedings of ACM MobiSys}, 2008.

\bibitem{OTGH10}
Femi~G. Olumofin, Piotr~K. Tysowski, Ian Goldberg, and Urs Hengartner.
\newblock {Achieving Efficient Query Privacy for Location Based Services}.
\newblock In {\em Privacy Enhancing Technologies}, pages 93--110, 2010.

\bibitem{PMX09}
Xiao Pan, Xiaofeng Meng, and Jianliang Xu.
\newblock Distortion-based anonymity for continuous queries in location-based
  mobile services.
\newblock In {\em GIS}, pages 256--265, 2009.

\bibitem{GDSB09}
Gabriel Ghinita, Maria~Luisa Damiani, Claudio Silvestri, and Elisa Bertino.
\newblock Preventing velocity-based linkage attacks in location-aware
  applications.
\newblock In {\em GIS}, pages 246--255, 2009.

\bibitem{DM09}
George Danezis and Prateek Mittal.
\newblock Sybilinfer: Detecting sybil nodes using social networks.
\newblock In {\em Proceedings of the Network and Distributed System Security
  Symposium (NDSS)}, 2009.

\bibitem{YKGF06}
Haifeng Yu, Michael Kaminsky, Phillip~B. Gibbons, and Abraham Flaxman.
\newblock Sybilguard: defending against sybil attacks via social networks.
\newblock {\em SIGCOMM Comput. Commun. Rev.}, 36:267--278, August 2006.

\bibitem{YGKX10}
Haifeng Yu, Phillip~B. Gibbons, Michael Kaminsky, and Feng Xiao.
\newblock Sybillimit: a near-optimal social network defense against sybil
  attacks.
\newblock {\em IEEE/ACM Trans. Netw.}, 18:885--898, June 2010.

\bibitem{CMS09}
Antonio Cutillo, Refik Molva, and Thorsten Strufe.
\newblock Safebook: Feasibility of transitive cooperation for privacy on a
  decentralized social network.
\newblock In {\em IEEE WOWMOM}, pages 1--6, 2009.

\bibitem{TSGW09}
A.~Tootoonchian, S.~Saroiu, Y.~Ganjali, and A.~Wolman.
\newblock {Lockr: Better Privacy for Social Networks}.
\newblock In {\em Proc. of ACM CoNEXT}, 2009.

\bibitem{BSB09}
Bobby~Bhattacharjee Randy~Baden, Neil~Spring.
\newblock Identifying close friends on the internet.
\newblock In {\em Hotnets}, 2009.

\bibitem{SZF10}
Jinyuan Sun, Xiaoyan Zhu, and Yuguang Fang.
\newblock A privacy-preserving scheme for online social networks with efficient
  revocation.
\newblock In {\em Proceedings of the 29th conference on Information
  communications}, INFOCOM'10, 2010.

\end{thebibliography}

\end{document}